\documentclass[journal]{IEEEtran}
\usepackage{amsmath,amssymb, amsthm}
\usepackage{algorithm, algorithmic}
\usepackage{graphicx}
\usepackage{multirow, booktabs}
\usepackage{multicol}
\usepackage{stfloats}
\usepackage{subfigure}
\usepackage{color}

\usepackage{enumerate, cite}
\usepackage{extarrows}
\usepackage{url}
\usepackage{epstopdf}
\usepackage{epsfig}
\usepackage{diagbox}
\usepackage{esint}
\newtheorem{proposition}{Proposition}

\def\A{\mathbf{A}}

\def\C{\mathbf{C}}

\def\I{\mathbf{I}}
\def\P{\mathbf{P}}
\def\Q{\mathbf{Q}}
\def\U{\mathbf{U}}

\def\a{\mathbf{a}}

\def\c{\mathbf{c}}
\def\d{\mathbf{d}}

\def\f{\mathbf{f}}
\def\g{\mathbf{g}}
\def\h{\mathbf{h}}
\def\q{\mathbf{q}}

\def\u{\mathbf{u}}

\def\w{\mathbf{w}}
\def\x{\mathbf{x}}
\def\y{\mathbf{y}}
\def\z{\mathbf{z}}
\def\aalpha{\boldsymbol{\alpha}}
\def\bbeta{\boldsymbol{\beta}}

\def\llambda{\boldsymbol{\lambda}}

\def\eeta{\boldsymbol{\eta}}
\def\LLambda{\boldsymbol{\Lambda}}
\def\TTheta{\boldsymbol{\Theta}}
\ifCLASSINFOpdf
\else
\fi

\begin{document}
\pdfoutput=1
\title{Wide-Beam Array Antenna Power Gain Maximization via ADMM Framework}

\author{Shiwen Lei,~\IEEEmembership{Member,~IEEE,}
        Jing Tian,~\IEEEmembership{Member,~IEEE,}
        Zhipeng Lin,~\IEEEmembership{Student Member,~IEEE,}
        Haoquan Hu,~
        Bo Chen,~
        Wei Yang,~
        Pu Tang,~
        and~Xiangdong Qiu

\thanks{S. Lei,~J. Tian,~Z. Lin, ~H. Hao, ~B. Chen, ~W. Yang and P. Tang are with the School of Electronic Science and Engineering, University of Electronic Science and Technology of China, Chengdu, Sichuan 611731, China. X. Qiu is with the Global Vision Media Technology Co., Ltd., Beijing, 100040, China}
\thanks{This work was partly supported by the National Natural Science Foundation of China under Grants U20B2043, 61801090 and 62001095.}
\thanks{Corresponding author: J. Tian (email: j.tian@uestc.edu.cn).}
}

\maketitle

\begin{abstract}
This paper proposes two algorithms to maximize the minimum array power gain in a wide-beam mainlobe by solving the power gain pattern synthesis (PGPS) problem with and without sidelobe constraints.
Firstly, the nonconvex PGPS problem is transformed into a nonconvex linear inequality optimization problem and then converted to an augmented Lagrangian problem by introducing auxiliary variables via the Alternating Direction Method of Multipliers (ADMM) framework.
Next, the original intractable problem is converted into a series of nonconvex and convex subproblems.
The nonconvex subproblems are solved by dividing their solution space into a finite set of smaller ones, in which the solution would be obtained pseudo-analytically.
In such a way, the proposed algorithms are superior to the existing PGPS-based ones as their convergence can be theoretically guaranteed with a lower computational burden.
Numerical examples with both isotropic element pattern (IEP) and active element pattern (AEP) arrays are simulated to show the effectiveness and superiority of the proposed algorithms by comparing with the related existing algorithms.
\end{abstract}

\IEEEpeerreviewmaketitle

\section{Introduction}

\IEEEPARstart{A}{rray} antenna has been widely applied in many modern electronic systems including radar, wireless communications (5G/satellite), sonar, remote sensing, etc., for its flexibility of synthesizing the desired radiation pattern \cite{LagerTSL_2009, DaiWLCW_2013, WangZHJ_2018, IupikovISCKC_2018}.
Basically, the array radiation patterns are divided into two categories including pencil beam pattern (PBP) \cite{YangZL_2013, ZhangHLZCL_2017, ButtazzoniV_2019, LeiHCTTQ_2019, Kedar_2020} and shaped beam pattern (SBP) \cite{ZhangHLZP_2018, LiLG_2018, ZhangHLYZZ_2019,RoccaPPM_2020, LiangFCZ_2020}.
The PBP is usually employed in target detection, direction of arrival estimation, point-to-point communication, etc., while the SBP is usually applied in regional coverage, wide-area searching, remote sensing and so on.
This paper focuses on the problem of maximizing the minimum power gain in a wide-beam mainlobe, which belongs to the type of SBP application.
Power gain in a given wide-beam mainlobe is usually expected to be high enough to guarantee the quality of signal receiving \cite{ChoiLL_2013, LimaMCFF_2015}.
One typical application utilizing such high gain wide-beam array is the satellite communication in motion, where the antenna is installed on a mobile platform such as train, airplane or other vehicle.
It, therefore, requires a wide-beam mainlobe with high power gain to guarantee the effective receiving of satellite signal during the moving of platform.


To synthesize a wide-beam with array antenna, the conventional methods are mainly SBPS-based, which try to narrow the mainlobe ripple level (MRL) with given the sidelobe level (SLL) or minimize the SLL with given the MRL.
For example, the SBPS problem is relaxed as a convex one, which therefore can obtain the optimal excitation weight by  solving a series of convex subproblems \cite{FuchsSM_2013, ZhangJZ_2020, QiBZ_2020}.
In addition, the alternating direction multiplier method (ADMM) framework is also utilized to simplify the SBPS optimization problem by updating the primal and dual variables iteratively \cite{YanYYZG_2017, LiangZSZ_2018, LiuJL_2020}.
Moreover, other methods such as semi-definite Relaxation (SDR)\cite{FuchsR_2016}, unitary matrix pencil (UMP)  \cite{ShenWL_2017, ZhuLLLL_2017} and pseudo-analytical approaches \cite{LiQZ_2017, QiL_2018}, etc., are applied to synthesize the desired SBP as well in past few years.
While these work try to obtain the SBP with a lower SLL \cite{GemechuCYK_2020}, or/and a narrower ripple \cite{FanLYCL_2019}, or/and a smaller excitation dynamic range ratio  \cite{LeiHTCTYQ_2019, FanLZSZ_2019},  the synthesis results cannot guarantee the highest power gain in a wide-beam mainlobe due to these SBPS-based methods focus on the pattern shape optimization.

To obtain the highest power gain in a wide-beam pattern, the intuitive idea is to select appropriate array excitation for power gain pattern optimization directly, in which the power gain pattern is set as cost function during the optimization. Such scenario is marked as power gain pattern synthesis (PGPS) optimization.
In our previous works, the PGPS problem is solved with an iterative algorithm \cite{LeiYHZCQ_2019} and a linear programming algorithm \cite{LeiHCTPX_2020}, respectively.
The numerical examples presented in these work validate that the PGPS-based algorithms are able to obtain higher minimum power gain in a wide-beam mainlobe than the SBPS-based algorithms.
It is also demonstrated that the PGPS-based iterative algorithm can obtain higher power gain than the linear programming one, while the latter one is SLL-controllable with a slight power gain degeneration.
These algorithms are further improved with a two-stage optimization algorithm \cite{ZhangJZ_2020e}, in which
the array excitation of current iteration  is updated inner an imperceptible variation of the previous iteration with a refined initialized guess.
Although these PGPS-based algorithms convert the original problem into (approximately) equivalent convex optimization that can be solved effectively with the readable tools such as CVX \cite{cvx} and SeDumi\cite{SeDumi}, the convergence of these PGPS-based algorithms cannot be theoretically guaranteed as their optimal solutions highly depend on the initialization parameters \cite{YangYHZD_2020}.
Though \cite{ZhangJZ_2020e} provides two well initialization guesses to obtain a higher power gain than the existing SBPS-based algorithms, its convergence still cannot be theoretically.
Besides, the computational burden of these PGPS-based algorithms would be heavy when the array size becomes large.
To overcome this shortcomings, this paper proposes two new algorithms that consider the scenarios without and with the SLL constraints to maximize the minimum power gain in a wide-beam mainlobe based on the ADMM framework.

The nonconvex PGPS problem is firstly reformulated into an augmented Lagrangian form by introducing auxiliary variables.
Then the reformulated problem is divided into a series of subproblems that can be solved iteratively to update the primal, dual and penalty variables.
It is worthy to mention that the primal-variables-related subproblems are intractable since these subproblems are still nonconvex.
In our algorithms, the solution space of the primal variables are divided into a finite set of smaller spaces, in which their solution can be obtained analytically.
The main contribution of this paper can be remarked as follows:
1). The nonconvex problem is reformulated into an augmented Lagrangian form to simplify the original problem. As a result, the intermediate variables can be updated separately.
2). The intractable subproblems are further simplified through dividing their solution spaces into finite smaller ones, in which the suboptimal solutions can be obtained analytically. Hence, the optimal solution can be achieved by minimizing the related cost function. Therefore, the nonconvex subproblems are pseudo-analytically solvable.
3). Based on the ADMM framework, the proposed algorithms are proved to be theoretically convergent.
As a sequel, the proposed algorithms can always obtain the suboptimal solution.

The remainders of this paper are organized as follows: the PGPS problem without/with the SLL constraints is formulated in Sec. II. The related proposed algorithms are deduced in Sec. III. The numerical examples are carried out and the remarks are discussed in Sec. IV while the conclusions are drawn in Sec. V.

\textit{Notations}: The superscripts $'H'$ and $'T'$ represents the Hermitian conjugate and conjugate of vector/matrix.
The superscript $'(k)'$ represents the $k^{th}$ iteration.
The boldface capital letter, such as $\A$, $\P$ and the boldface lower case letter, such as $\c$, $\g$, represent matrix and vector, respectively.
$'\odot'$ represents the element-wise multiplication operator.
$'\angle\x'$ represents the angular vector of vector $\x$.
$'||\x||_1'$ and $'||\x||_2'$ represent the $\ell_1$ norm and $\ell_2$ norm of vector $\x$, respectively.

\section{Problem Formulation}
For simplicity, a linear array is discussed in this work and the related results can easily be extended to the planar array.
For any $N$-elements arbitrarily distributed linear array with the element position $r_n~(n=1, \cdots, N)$, the array far-field electric field at $\theta$ direction can be expressed as the sum of those of elements:
\begin{align}
E(\theta)=\sum_{n=1}^N\omega_n\exp(2j\pi\kappa r_n\sin(\theta)\triangleq\w^H\a
\end{align}
where $\w=[\omega_1, \cdots, \omega_N]^T$, $\a=[a_1(\theta), \cdots, a_N(\theta)]^T$ are the array weight and the array factor with $a_n(\theta)=\exp(2j\pi\kappa r_n\sin(\theta))$, $\theta\in[-\frac{\pi}{2}, \frac{\pi}{2}]$ and $\kappa=\frac{\lambda}{2\pi}$ being the wavenumber.
The effective radiation beam pattern considering the element total efficiency can be expressed as \cite{LeiHCTPX_2020}:
\begin{align}
P(\theta)=(\w\odot\sqrt{\eeta})^H\a(\theta)
\end{align}
where $\eeta=[\eta_1, \cdots, \eta_N]^T$ is the element total efficiency including the reflection, conduction and dielectric efficiencies \cite{balanis_2016}.
In practice, the total efficiency can be affected by both beam direction and array weight. Since the relationship among the beam direction, the array weight and the total efficiency $\eeta$ can be determined experimentally, $\eeta$ is assumed as known $\textit{a priori}$ in this paper.
To simplify the expression, redefine $\w=\w\odot\eta$ and the array power gain can be written as
\begin{align}
G(\theta)=\frac{4\pi|P(\theta)|^2}{2\pi\oint|P(\theta)|^2\cos(\theta)d\theta}=\frac{2\w^H\A_\theta\w}{\w^H\A\w}
\end{align}
where $\A_{\theta}=\a^H(\theta)\a(\theta)$ and $\A=\oint\a^H(\theta)\a(\theta)\cos(\theta)d\theta$.
As a result, $\A_{\theta}$ and $\A$ are both positive semi-definite matrices.
Given the array structure, parameters $\A_{\theta}$ and $\A$, can be computed accordingly.

To generate a wide-beam mainlobe with optimal power gain through optimizing the array weight $\w$, we need to solve the optimization problem in the following forms:
\begin{itemize}
\item [1).]without SLL constraint:
\begin{align} \label{eq:PGPS_without_SLL}
\begin{split}
&~~~~~~~\max_{\w} ~G_0\\
&s.t. ~~~\frac{\w^H\A_{\theta}\w}{\w^H\A\w} \geq G_0, ~~\theta\in\TTheta_\text{ML}\\
\end{split}
\end{align}
\item [2).] with SLL constraint:
\begin{align}\label{eq:PGPS_with_SLL}
\begin{split}
&~~~~~~~\max_{\w} ~G_0\\
&s.t.~~~ \frac{\w^H\A_{\theta}\w}{\w^H\A\w} \geq G_0,~~~\theta\in\TTheta_\text{ML}\\
&~~~~~~~ \frac{\w^H\A_{\theta}\w}{\w^H\A\w} \leq \gamma G_0,~~\theta\in\TTheta_\text{SL}
\end{split}
\end{align}
\end{itemize}
where $\TTheta_\text{ML}$ and $\TTheta_\text{SL}$ are the wide-beam mainlobe region and sidelobe region, respectively.
$\gamma$ is the desired peak SLL.

\section{Proposed Algorithms for Power Gain Maximization}
In this section, two effective algorithms for solving problem (\ref{eq:PGPS_without_SLL}) and (\ref{eq:PGPS_with_SLL}) are proposed based on the ADMM framework with a guaranteed convergence property.
In physics, $\w^H\A\w$ denotes the transmission or reception total power of the electromagnetic wave, which is greater than zero for any $\w\neq0$, i.e., $\w^H\A\w>0$. As a result, $\A$ is a positive definite matrix and there exists a positive definite matrix $\C$ so that $\A=\C^H\C$.
Rewrite the mainlobe region $\TTheta_\text{ML}$ as discrete $\text{L}_\text{ML}$ angular and denote $\x=\C\w$, $\P_l=\C^{-H}\A(\theta_l)\C^{-1}$, $l=1, \cdots, \text{L}_\text{ML}$ with $\theta_l\in\TTheta_{\text{L}_\text{ML}}$, $\w^H\A(\theta_l)\w$ can be converted into the $\ell_2$ norm of vector $\a^H(\theta_l)\C^{-1}\x$, i.e., $\w^H\A(\theta_l)\w=\x^H\P_l\x=||\a^H(\theta_l)\C^{-1}\x||^2_2$.


\subsection{Power Gain Maximization Without SLL Constraint}
The objective function in (\ref{eq:PGPS_without_SLL}) (and (\ref{eq:PGPS_with_SLL})) can be converted into following form:
\begin{align}
\max_{\w}~G_0~\Leftrightarrow~\min_{\w}~-G_0
\end{align}

Since the zooming of $\x$ does not affect the final solution to the problem, $||\x||_2$ can be arbitrary real number.
Without loss of generality, set $||\x||_2=1$.
Problem (\ref{eq:PGPS_without_SLL}) can be rewritten as
\begin{align} \label{eq:PGPS_without_SLL1}
\begin{split}
&~~~~~~~\min_{\x} ~-G_0\\
&s.t. ~~~||\a^H(\theta_l)\C^{-1}\x||_2^2=G_l,~~ l=1,\cdots, \text{L}_\text{ML}\\
&~~~~~~~ G_l\geq G_0,\\
&~~~~~~~||\x||_2=1
\end{split}
\end{align}

Define $\c_l=\C^{-H}\a(\theta_l)$, $g_0=\sqrt{G}_0$ and $|g_l|=\sqrt{G_l}$, problem (\ref{eq:PGPS_without_SLL1}) can be simplified as
\begin{align}
\begin{split}
&~~~~~\min_{\x} ~-g_0\\
&s.t. ~~~\c_l^H\x=g_l,~~ l=1,\cdots, \text{L}_\text{ML}\\
&~~~~~~~ |g_l|\geq g_0\\
&~~~~~~~||\x||_2=1
\end{split}
\end{align}

The ADMM framework can fast converge by blending the decomposability of dual ascent with superior convergence properties of the multiplier method \cite{BoydPCPE_2011}.
With the aid of this framework, the augmented Lagrangian form of the above optimization problem can be denoted as:
\begin{align} \label{eq:PGPS_without_SLL_Lag}
\begin{split}
&~~\min_{\x, g_0, \g,\u} L_{\rho}=-g_0+\frac{1}{2\rho}||\P^H\x-\g+\rho\u||_2^2\\
&s.t.~~~~~ |g_l|\geq g_0\\
&~~~~~~~||\x||_2=1
\end{split}
\end{align}
where  $\P=[\c_1^H, \cdots, \c_{\text{L}_\text{ML}}^H]^H$ and $\g=[g_1, \cdots, g_{\text{L}_\text{ML}}]^T$.
$\u=[u_1, \cdots, u_{\text{L}_\text{ML}}]^T$ are dual variables.
$\rho$ is the predetermined tuning parameter.
In such a form, the above problem can be decomposed into three subproblems, which updates the variables $\x$, $\{g_0,\g\}$ and $\u$ in (\ref{eq:PGPS_without_SLL_Lag}) iteratively as follows:
\begin{subequations}
\begin{align}
\begin{split}
\{g_0^{(k+1)}, \g^{(k+1)}\}&=\min_{g_0, \g} L_{\rho}\left(\x^{(k)}, g_0, \g, \u^{(k)}\right)\label{eq:PGPS_without_SLL_gg}\\
s.t.~~&|g_l|\geq g_0, ~~l=1,\cdots,\text{L}_\text{ML}
\end{split}
\end{align}
\begin{align}
\begin{split}
\x^{(k+1)}&=\min_{\x} L_{\rho}\left(\x, g_0^{(k+1)}, \g^{(k+1)}, \u^{(k)}\right)\label{eq:PGPS_without_SLL_x}\\
&s.t.~~||\x||_2=1
\end{split}
\end{align}
\begin{align}
\u^{(k+1)}&=\u^{(k)}+\frac{\P^H\x^{(k+1)}-\g^{(k+1)}}{\rho}\label{eq:PGPS_without_SLL_u}
\end{align}
\end{subequations}

The detailed procedures of solving these subproblems are as follows:
\subsubsection{Solving problem (\ref{eq:PGPS_without_SLL_gg})}
Problem (\ref{eq:PGPS_without_SLL_gg}) can be expressed as
\begin{align}
\begin{split}
&\min_{g_0, \g} -g_0+\frac{1}{2\rho}||\P^H\x^{(k)}-\g+\rho\u^{(k)}||_2^2\\
s.t.~~&|g_l|\geq g_0, ~~l=1,\cdots,\text{L}_\text{ML}
\end{split}
\end{align}
where $\P^H\x^{(k)}+\rho\u$ is a $\text{L}_\text{ML}\times 1$ dimensional vector. Hence, the absolute values of its elements divide the domain $[0, \infty]$ into $\text{L}_\text{ML}+1$ sub-domains.
If $\{g_0, \g\}$ can be easily decided in each sub-domain, the optimal solution to this problem is the one that minimizes the cost function $L_{\rho}(g_0, \g)$ with $\text{L}_\text{ML}+1$ suboptimal solutions .

Let $\y=\P^H\x^{(k)}+\rho\u^{(k)}$,
rearrange $\y$ in the form of ascending order with respect to (w.r.t) its elements modulus.
A new vector marked $\tilde{\y}=[\tilde{y}_1, \cdots, \tilde{y}_{\text{L}_\text{ML}}]^T$ with $|\tilde{y}_{l_1}|\geq|\tilde{y}_{l_2}|$ for $1\leq l_1\leq l_2\leq\text{L}_\text{ML}$ is defined.
To obtain the optimal $\{g_0, \g\}$, the optimization problem (\ref{eq:PGPS_without_SLL_gg}) can be rewritten as
\begin{align}
\begin{split}
&\min_{g_0, \g}~~L_{\rho}(g_0, \g)=-g_0+\frac{1}{2\rho}||\tilde{\y}-\g||_2^2\\
&~~s.t.~~~ 0 < g_0\leq |g_l|, ~~l = 1, \cdots, \text{L}_\text{ML}
\end{split}
\end{align}

Since $g_0>0$, its optimal value that minimizes $L_{\rho}(g_0, \g)$ locates in one of the following domains $(0, |\tilde{y}_1|]$, $\cdots$, $[|\tilde{y}_{l-1}|, |\tilde{y}_l|]$, $\dots$, $[|\tilde{y}_{\text{L}_\text{ML}}|,+\infty]$.
The corresponding minimum value of $L_{\rho}(g_0, \g)$ can be computed as follows:
\begin{itemize}
\item If $g_0\in(0, |\tilde{y}_1|]$
\begin{align} \label{eq:PGPS_without_SLL_gg1}
L^{\min}_{\rho}(g_0, \g)=-|\tilde{y}_1|+\frac{1}{2\rho}||\tilde{\y}-\tilde{\y}||_2^2=-|\tilde{y}_1|
\end{align}
with
\begin{align}\label{eq:PGPS_without_SLL_gg_value1}
\begin{split}
\begin{cases}
g_0&=|\tilde{y}_1|\\
\g&=\tilde{\y}
\end{cases}
\end{split}
\end{align}
\item If $g_0\in(|\tilde{y}_{l-1}|, |\tilde{y}_l|]$ for $ l=2, \cdots, \text{L}_\text{ML}$
\begin{align}\label{eq:PGPS_without_SLL_gg2}
\begin{split}
&L^{\min}_{\rho}(g_0, \g)\\
&=\frac{ l}{2\rho} g_0^2-\left(1+\frac{1}{\rho}\sum_{m=1}^l|\tilde{y}_m|\right)g_0+\frac{1}{2\rho}\sum_{m=1}^l|\tilde{y}_m|^2
\end{split}
\end{align}
Let $c_l = \frac{\rho+\sum_{m=1}^l|\tilde{y}_m|}{l}$, the optimal $g_0$ and $\g$ can be obtained via the following expressions:
\begin{align}\label{eq:PGPS_without_SLL_gg_value2}
\begin{split}
\begin{cases}
g_0&=\begin{cases}
c_l, ~~ |\tilde{y}_{l-1}|\leq c_l\leq|\tilde{y}_l|\\
|\tilde{y}_l|,~~c_l>|\tilde{y}_l|\\
|\tilde{y}_{l-1}|,~~c_l<|\tilde{y}_{l-1}|
\end{cases}\\
\g&=\max\{g_0, |\tilde{\y}|\}\odot\exp(1j\angle{\tilde{\y}})
\end{cases}
\end{split}
\end{align}
\item If $g_0\geq|\tilde{y}_{\text{L}_\text{ML}}|$
\begin{align}\label{eq:PGPS_without_SLL_gg3}
\begin{split}
L^{\min}_{\rho}(g_0, \g)&=\frac{\text{L}_\text{ML}}{2\rho} g_0^2-\left(1+\frac{\text{L}_\text{ML}}{\rho}||\tilde{\y}_m||_1\right)g_0\\
&+\frac{\text{L}_\text{ML}}{2\rho}||\tilde{\y}||_1^2
\end{split}
\end{align}
Let $c_{\text{L}_\text{ML}} = \frac{\rho+||\tilde{\y}||_1}{\text{L}_\text{ML}}$, the optimal $\{g_0, \g\}$ can be obtained via the following expressions:
\begin{align}\label{eq:PGPS_without_SLL_gg_value3}
\begin{split}
\begin{cases}
g_0&=\max\left\{c_{\text{L}_\text{ML}},|\tilde{y}_{\text{L}_\text{ML}}|\right\}\\
g_m&=g_0\exp(1j\angle(\tilde{\y}))
\end{cases}
\end{split}
\end{align}
\end{itemize}

As a result, the optimal $\{g_0,\g\}$ to (\ref{eq:PGPS_without_SLL_gg}) is the one that obtains the minimum $L^{\min}_{\rho}(g_0, \g)$ described in (\ref{eq:PGPS_without_SLL_gg1}), (\ref{eq:PGPS_without_SLL_gg2}) and (\ref{eq:PGPS_without_SLL_gg3}).

\subsubsection{Solving problem (\ref{eq:PGPS_without_SLL_x})}
Omit the irrelevant term, problem (\ref{eq:PGPS_without_SLL_x}) can be expressed as
\begin{align}\label{eq:a1_sub_x}
\begin{split}
&\min_{\x} \frac{1}{2\rho}||\P^H\x-\g^{(k+1)}+\rho\u^{(k)}||_2^2\\
&s.t.~~||\x||_2=1
\end{split}
\end{align}

This problem is equivalent to minimizing $||\P^H\x-\g^{(k+1)}+\rho\u^{(k)}||_2^2$ under constraint $||\x||_2=1$.
A pseudo-analytical algorithm with linear computation complexity is proposed for this optimization problem as follows.
Let $\d=\g^{(k+1)}-\rho\u^{(k)}$, the optimal $\x$ that minimizes the objective function in the optimization problem (\ref{eq:a1_sub_x}) can be achieved by minimizing the following function:
\begin{align}
f(\x) = ||\P^H\x-\d||^2_2,
\end{align}
with the constraint $||\x||_2=1$.

Define
\begin{subequations}
\begin{align}
\tilde{\P} &= \begin{bmatrix}
\text{real}(\P) & -\text{imag}(\P)\\
\text{imag}(\P) &\text{real}(\P)
\end{bmatrix}\label{eq:P_tilde_matrix}\\
\tilde{\x}&=\begin{bmatrix}
\text{real}(\x)\\
\text{imag}(\x)
\end{bmatrix}\label{eq:x_tilde_vector}\\
\tilde{\d}&=\begin{bmatrix}
\text{real}(\d)\\
\text{imag}(\d)
\end{bmatrix}\label{eq:d_tilde_vector}
\end{align}
\end{subequations}

Function $f(\x)$ can be rewritten as
\begin{align} \label{eq:f3_real_value}
\begin{split}
f(\tilde{\x}) &= (\tilde{\P}^T\tilde{\x}-\tilde{\d})^T(\tilde{\P}^T\tilde{\x}-\tilde{\d})\\
&=\tilde{\x}^T\tilde{\P}\tilde{\P}^T\tilde{\x}-2\tilde{\x}^T\tilde{\P}\tilde{\d}+\tilde{\d}^T\tilde{\d}
\end{split}
\end{align}

Let $\llambda=[\lambda_1,\cdots, \lambda_N]^T$ and $\U=[\u_1,\cdots, \u_N]$ denote the eigenvalues and the corresponding unit eigenvectors of matrix $\tilde{\P}\tilde{\P}$ respectively.
As a result, there must exist such vectors $\bbeta=[\beta_1, \cdots, \beta_N]^T$ and $\aalpha=[\alpha_1, \cdots, \alpha_N]^T$ that
\begin{subequations}\label{eq:Pd_x_def}
\begin{align}
\tilde{\P}\tilde{\d}&=\U\bbeta\label{eq:Pd_x_def1}\\
\tilde{\x}&=\U\aalpha\label{eq:Pd_x_def2}
\end{align}
\end{subequations}

Define $\LLambda=\text{diag}\{\llambda\}$, substitute (\ref{eq:Pd_x_def}) into (\ref{eq:f3_real_value}) gives
\begin{align}
f(\aalpha, \bbeta) = \aalpha^T\LLambda\aalpha - 2\bbeta^T\aalpha+\tilde{\d}^T\tilde{\d}
\end{align}

With the constraint $||\tilde{\x}||_2=1$, a new Lagrangian function related with $f(\aalpha, \bbeta)$ can be defined as
\begin{align}\label{eq:f3_alpha_beta}
\begin{split}
f_{\nu}(\aalpha, \bbeta)&=f(\aalpha, \bbeta) + \nu(1-\tilde{\x}^T\tilde{\x})\\
&=f(\aalpha, \bbeta) +\nu(1-\aalpha^T\aalpha)
\end{split}
\end{align}

Let the partial derivative of $f_{\nu}(\aalpha, \bbeta)$ w.r.t $\aalpha$ equal to zero, it obtains
\begin{align} \label{eq:allpha}
\aalpha=(\LLambda-\nu\I)^{-1}\bbeta
\end{align}

Substitute ($\ref{eq:allpha}$) into the constraint $||\tilde{\x}||_2=1$, it yields
\begin{align}\label{eq:nu_value}
\begin{split}
\aalpha^T\aalpha&=\bbeta^T\left(\LLambda-\nu\I\right)^{-2}\bbeta\\
&=\sum_{n=1}^N\left(\frac{\beta_n}{\nu-\lambda_n}\right)^2=1
\end{split}
\end{align}

Variable $\nu$ in above equation can be solved by searching a small region as follows with bisection strategy:
\begin{align}
\begin{split}
&\nu_{\min}\in\left[\min\left\{\lambda_n-\sqrt{N}|\beta_n|\right\}, \right.\\
&~~~\left.\min\left\{\min\left\{\lambda_n-|\beta_n|\right\},\max\left\{\lambda_n-\sqrt{N}|\beta_n|\right\}\right\}\right]
\end{split}
\end{align}

The obtainment of the above region is referred to Appendix \ref{app:a} for details, then $\x$ can be updated with the aid of (\ref{eq:x_tilde_vector}) and (\ref{eq:Pd_x_def2}).
Consequently, the solving processes of problem (\ref{eq:PGPS_without_SLL}) can be summarized as in \textbf{Algorithm \ref{alg:PGM-WoSC}}.
For simplicity, mark the proposed algorithm as 'PGM-WoSC', which is the abbreviation of 'Power Gain Maximization Without SLL Constraint'.
The PGM-WoSC is terminated either when the predetermined maximum iteration is reached or when the predetermined estimation accuracy of $\g$ is reached.
\begin{algorithm}[ht]
\caption{The PGM-WoSC algorithm for problem (\ref{eq:PGPS_without_SLL_Lag})}
\begin{algorithmic}[1] \label{alg:PGM-WoSC}
\STATE initialize $\rho\in(1, 10000)$, $\x^{(0)}=\textbf{0}_{\text{L}_\text{ML}\times1}$, $\u^{(0)}=\textbf{0}_{\text{L}_\text{ML}\times1}$
\WHILE{$ k<\textit{IterMax}$~$\text{or}$~$\max\{|\P^H\x-\g|\}>10^{-4}$}
\STATE update $\{g_0^{(k)}, \g^{(k)}\}$ by minimizing $L_{\rho}(g_0, \g)$ in  (\ref{eq:PGPS_without_SLL_gg1}), (\ref{eq:PGPS_without_SLL_gg2}) and (\ref{eq:PGPS_without_SLL_gg3}) with (\ref{eq:PGPS_without_SLL_gg_value1}), (\ref{eq:PGPS_without_SLL_gg_value2}) and (\ref{eq:PGPS_without_SLL_gg_value3}), respectively. \label{alg1:update_gg}
\STATE update $\x^{(k)}$  by solving (\ref{eq:Pd_x_def}) with the estimated $\aalpha$ via (\ref{eq:nu_value}) \label{alg1:update_x}
\STATE update $\u^{(k)}$ with (\ref{eq:PGPS_without_SLL_u}) \label{alg1:update_u}
\STATE $k=k+1$
\ENDWHILE
\STATE return array weight $\w=\C^{-1}\x$.
\end{algorithmic}
\end{algorithm}

\subsection{Power Gain Maximization With SLL Constraint}
In many applications, the array antenna is often required to have high mainlobe power gain with controllable SLL.
To satisfy this practical demand, the optimization problem comes out to be the form of (\ref{eq:PGPS_with_SLL}).
Define $\q_s=\C^{-H}\a(\theta_s)$ and $\P_s=\C^{-H}\A(\theta_s)\C~(s=1, \cdots, \text{L}_\text{SL})$ for $\theta_s\in\TTheta_{\text{L}_\text{SL}}$, problem (\ref{eq:PGPS_with_SLL}) can be reformulated as
\begin{align}
\begin{split}
&~~\min_{\x} ~-g_0\\
&s.t.~~ \P^H\x = \g\\
&~~~~~~\Q^H\x = \h\\
&~~~~~~~~~~g_0 \leq|g_l|\\
&~~~~~~\sqrt{\gamma}g_0 \geq |h_s|\\
&~~~~~~||\x||_2=1
\end{split}
\end{align}
where $\Q=[\q_1^H, \cdots, \q_{\text{L}_\text{SL}}^H]^H$, $\h=[h_1, \cdots, h_{\text{L}_\text{SL}}]^T$.

The augmented Lagrangian problem form of above optimization problem can be written as:
\begin{align}\label{eq:PGPS_with_SLL_Lag}
\begin{split}
&~~\min_{\x,g_0, \g, \h,\u_1, \u_2} ~L_{\rho_1, \rho_2}=-g_0+\frac{1}{2\rho_1}||\P^H\x-\g+\rho_1\u_1||_2^2\\
&~~~~~~~~~~~~~~~~~~~~~~~~~~~~~~~~~~~+\frac{1}{2\rho_2}||\Q^H\x-\h+\rho_2\u_2||_2^2\\
&~~~~~~s.t.~~~~~~g_0 \leq|g_l|\\
&~~~~~~~~~~~~\sqrt{\gamma}g_0 \geq |h_s|\\
&~~~~~~~~~~~~||\x||_2=1
\end{split}
\end{align}
where $\rho_1$ and $\rho_2$ are the predetermined tuning parameters. $\u_1$ and $\u_2$ are dual variables.
The above problem can be decomposed into four subproblems, which updates the
variables $\x$, $\{g_0, \g, \h\}$, $\u_1$ and $\u_2$ in (\ref{eq:PGPS_with_SLL_Lag}) iteratively as follows:

\begin{subequations}\label{eq:PGPS_with_SLL_iter}
\begin{align}\label{eq:PGPS_with_SLL_ggh}
\begin{split}
\{g_0^{(k+1)}, &\g^{(k+1)}, \h^{(k+1)}\}\\
&=\min_{g_0, \g, \h}L_{\rho_1, \rho_2}\left(\x^{(k)}, g_0, \g, \h, \u_1^{(k)}, \u_2^{(k)}\right)\\
&~~~~~~~s.t.~~g_0\leq|g_l|\\
&~~~~~~~~~\sqrt{\gamma}g_0\geq|h_s|
\end{split}
\end{align}
\begin{align}\label{eq:PGPS_with_SLL_x}
\begin{split}
\x^{(k+1)}&=\min_{\x}L_{\rho_1, \rho_2}\left(\x, g_0^{(k+1)}, \g^{(k+1)}, \h^{(k+1)},
\u_1^{(k)}, \u_2^{(k)}\right)\\
&~~~~~~s.t.~~||\x||_2=1
\end{split}
\end{align}
\begin{align}\label{eq:PGPS_with_SLL_u1}
\u_1^{(k+1)}&=\u_1^{(k)} + \frac{\P^H\x^{(k+1)}-\g^{(k+1)}}{\rho_1}
\end{align}
\begin{align}\label{eq:PGPS_with_SLL_u2}
\u_2^{(k+1)}&=\u_2^{(k)} + \frac{\Q^H\x^{(k+1)}-\h^{(k+1)}}{\rho_2}
\end{align}
\end{subequations}

It is noted that $g_0$ in problem (\ref{eq:PGPS_without_SLL_gg}) is only related  to the undetermined variable $\g$. However, in problem (\ref{eq:PGPS_with_SLL_ggh}) $g_0$ is related to the undetermined variables $\g$ and $\h$, which therefore increases the problem complexity. A similar strategy is applied to solve (\ref{eq:PGPS_with_SLL_ggh}) and problem (\ref{eq:PGPS_with_SLL_x}) can be solved via the similar processes for problem (\ref{eq:PGPS_without_SLL_x}) with a series of matrices and vectors transformations.
The detailed procedures of solving these subproblems are as follows:
\subsubsection{Solving problem (\ref{eq:PGPS_with_SLL_ggh})}
Problem (\ref{eq:PGPS_with_SLL_ggh}) can be expressed as
\begin{align}
\begin{split}
&\min_{g_0, \g, \h} -g_0+\frac{1}{2\rho_1}||\P^H\x^{(k)}-\g+\rho_1\u_1^{(k)}||_2^2\\
&~~~~~~~~~~~~+\frac{1}{2\rho_2}||\Q^H\x^{(k)}-\h+\rho_2\u_2^{(k)}||_2^2\\
&~~~~~~~s.t.~~g_0\leq|g_l|\\
&~~~~~~~~~\sqrt{\gamma}g_0\geq|h_s|
\end{split}
\end{align}

where $\z_1=\P^H\x^{(k)}+\rho_1\u_1^{(k)}$ and $\z_2=\Q^H\x^{(k)}+\rho_2\u_2^{(k)}$ are $\text{L}_\text{ML}\times1$ and $\text{L}_\text{SL}\times1$ dimensional vectors, respectively.
Hence, vector $\z$ composed of the elements of $[|\z_1|, |\frac{\z_2}{\sqrt{\gamma}}|]$ in an ascending order is $(\text{L}_\text{ML}+\text{L}_\text{SL})\times1$ dimensional.
Its elements, hereby, divide the domain $[0, \infty]$ into $\text{L}_\text{ML}+\text{L}_\text{SL}+1$ sub-domains.
Hence, the optimal $\{g_0, \g, \h\}$ is one of the suboptimal $\{g_0, \g, \h\}$ to each sub-domain, which minimizes the related objective function $L_{\rho_1, \rho_2}(g_0, \g, \h)$.
Specifically, the solution of the following optimization problem would yield the optimal $\{g_0, \g, \h\}$:
\begin{align}\label{eq:PGPS_with_SLL_ggh_lag}
\begin{split}
&\min_{g_0, \g, \h} L_{\rho_1, \rho_2}(g_0, \g, \h) = -g_0 +\frac{1}{2\rho_1}||\z_1-\g||_2^2 \\
&~~~~~~~~~~~~~~~~~~~~~~~~~~~~~+\frac{1}{2\rho_2}||\z_2-\h||_2^2\\
&s.t. ~~~~~~~~~~ g_0\leq|g_l|,~l=1,\cdots, \text{L}_\text{ML}\\
&~~~~~~~~~~\sqrt{\gamma}g_0\geq|h_s|, ~s=1, \cdots, \text{L}_\text{SL}
\end{split}
\end{align}

Rearrange $\z_1$ and $\z_2$ in an ascending order w.r.t. their modulus respectively to create two new vectors, i.e., $\tilde{\z}_1$ and $\tilde{\z}_2$, and define a new vector $\z$ by rearranging vectors $[|\z_1|, \frac{|\z_2|}{\sqrt{\gamma}}]$ in an ascending order, the elements of $\z$ divides the domain $(0, \infty)$ into $\text{L}_\text{ML}+\text{L}_\text{SL}+1$ sub-domains, which are $(0, z_1]$, $\cdots$, $[z_{m-1}, z_m]$, $\cdots$, $[z_{\text{L}_\text{ML}+\text{L}_\text{SL}+1}, \infty)$ with $m=1, \cdots, \text{L}_\text{ML}+\text{L}_\text{SL}+1$.
As the optimal $g_0$ locates in one of these sub-domains, problem (\ref{eq:PGPS_with_SLL_ggh_lag}) can be solved with the following way:
\begin{itemize}
\item If $g_0\in(0, z_1$]
\begin{align}\label{eq:PGPS_with_SLL_gg1}
\begin{split}
&L_{\rho_1, \rho_2}^{\min}(g_0, \g, \h)\\
&~~~~~=\frac{\text{L}_\text{SL}\gamma}{2\rho_2}g_0^2-(1+\frac{\sqrt{\gamma}}{\rho_2}||\tilde{\z}_2||_1)g_0+\frac{||\tilde{\z}_2||_2^2}{2\rho_2}
\end{split}
\end{align}
where $\tilde{z}_{2, m_2}$ represents the $m_2^{th}$ element of $\tilde{z}_2$.
Let $d_1=\frac{\rho_2+\sqrt{\gamma}||\z_2||_1}{\text{L}_\text{SL}\gamma}$, the optimal $\{g_0, \g, \h\}$ can be computed via the following expressions:
\begin{align}\label{eq:PGPS_with_SLL_gg_value1}
\begin{split}
\begin{cases}
g_0&=\min\{z_1, d_1\}\\
\g&=\max\{g_0, |\tilde{\z}_1|\}\odot\exp(1j\angle{\tilde{\z}_1})\\
\h&=\min\{\sqrt{\gamma}g_0, |\tilde{\z}_2|\}\odot\exp(1j\angle{\tilde{\z}_2})
\end{cases}
\end{split}
\end{align}
\item If $g_0\in[z_{m-1}, z_{m}]$ for $m=2,\cdots,\text{L}_\text{ML}+\text{L}_\text{SL}+1$
\begin{align}\label{eq:PGPS_with_SLL_gg2}
\begin{split}
&L^{\min}_{\rho_1, \rho_2}(g_0, \g, \h)
=\left(\frac{l}{2\rho_1}+\frac{(\text{L}_{\text{SL}}-s+1)\gamma}{2\rho_2}\right)g_0^2\\
&~~-\left(1+\frac{1}{\rho_1}\sum_{m_1=1}^{l}|\tilde{z}_{1,m_1}|+\frac{\sqrt{\gamma}}{\rho_2}\sum_{m_2=s}^{\text{L}_\text{SL}}|\tilde{z}_{2,m_2}|\right)g_0\\
&~~+\left(\frac{1}{2\rho_1}\sum_{m_1=1}^{l}|\tilde{z}_{1,m_1}|^2+\frac{1}{2\rho_2}\sum_{m_2=s}^{\text{L}_\text{SL}}|\tilde{z}_{2,m_2}|^2\right)
\end{split}
\end{align}
where $\tilde{z}_{1,m_1}$ represents the $m_1^{th}$ element of $\tilde{\z}_1$. $l$ and $s$ are indexes that satisfy $\tilde{z}_{1,l}\geq z_m$ and $\tilde{z}_{2,s}\geq z_m$.
Let $d_m=\frac{\rho_1\rho_2+\rho_2\sum_{m_1=1}^{l}|\tilde{z}_{1,m_1}|+\rho_1\sqrt{\gamma}\sum_{m_2=s}^{\text{L}_\text{SL}}|\tilde{z}_{2,m_2}|}{\rho_2l+\rho_1(\text{L}_\text{SL}-s+1)\gamma}$, the optimal $\{g_0, \g, \h\}$ can be computed via the following expressions:
\begin{align}\label{eq:PGPS_with_SLL_gg_value2}
\begin{split}
\begin{cases}
g_0&=\begin{cases}
d_m, ~z_{m-1}\leq d_m \leq z_{m}\\
z_m, ~d_m>z_m\\
z_{m-1},~d_m<z_{m-1}\\
\end{cases}\\
\g&=\max\{g_0, |\tilde{\z}_1|\}\odot\exp(1j\angle{\tilde{\z}_1})\\
\h&=\min\{\sqrt{\gamma}g_0, |\tilde{\z}_2|\}\odot\exp(1j\angle{\tilde{\z}_2})
\end{cases}
\end{split}
\end{align}
\item If $g_0\in[z_{\text{L}_\text{ML}+\text{L}_\text{SL}+1}, \infty)$
\begin{align}\label{eq:PGPS_with_SLL_gg3}
\begin{split}
&L^{\min}_{\rho_1, \rho_2}(g_0, \g, \h)\\
&~~=\left(\frac{\text{L}_\text{ML}}{2\rho_1}\right)g_0^2-\left(1+\frac{1}{\rho_1}||\tilde{\z}||_1\right)g_0+||\tilde{\z}_1||_1
\end{split}
\end{align}
Define $d_{\text{L}_\text{ML}+\text{L}_\text{SL}+1}=\frac{\rho_1+||\tilde{\z}_1||_1}{\text{L}_\text{ML}}$, the optimal $\{g_0, \g, \h\}$ can be computed via the following expressions:
\begin{align}\label{eq:PGPS_with_SLL_gg_value3}
\begin{split}
\begin{cases}
g_0&=\max\{d_{\text{L}_\text{ML}+\text{L}_\text{SL}+1}, z_{\text{L}_\text{ML}+\text{L}_\text{SL}+1}\}\\
\g&=g_0\exp(1j\angle{\tilde{\z}_1})\\
\h&=\tilde{\z}_2
\end{cases}
\end{split}
\end{align}
\end{itemize}

Among such $\text{L}_\text{ML}+\text{L}_\text{SL}+1$ suboptimal $\{g_0, \g, \h\}$, the one that minimizes the objective function $L_{\rho_1, \rho_2}(g_0, \g, \h)$ is the final result of problem (\ref{eq:PGPS_with_SLL_ggh_lag}), which also provides the optimal solution to problem (\ref{eq:PGPS_with_SLL_ggh}).

\subsubsection{Solving problem (\ref{eq:PGPS_with_SLL_x})}
Problem (\ref{eq:PGPS_with_SLL_x}) can be expressed as
\begin{align}
\begin{split}
&\min_{\x} ~~\frac{1}{2\rho_1}||\P^H\x-\g^{(k+1)}+\rho_1\u_1^{(k)}||_2^2\\
&~~~~+\frac{1}{2\rho_2}||\Q^H\x-\h^{(k+1)}+\rho_2\u_2^{(k)}||_2^2\\
&~~~~~~~s.t.~~||\x||_2=1
\end{split}
\end{align}

Let $\d_1=\g^{(k+1)}-\u_1^{(k)}$ and $\d_2=\h^{(k+1)}-\u_2^{(k)}$, the optimal $\x$ that minimizes function $L_{\rho_1, \rho_2}(\x)$ in problem (\ref{eq:PGPS_with_SLL_x}) can be achieved by minimizing the following function:
\begin{align}
f_1(\x)=||\P^H\x-\d_1||_2^2 + ||\Q^H\x-\d_2||_2^2
\end{align}
with constraint $||\x||_2=1$.
Define
\begin{subequations}
\begin{align}
\tilde{\Q}&=\begin{bmatrix}
\text{real}(\Q) & -\text{imag}(\Q) \\
\text{imag}(\Q) &\text{real}(\Q) \\
\end{bmatrix}\\
\tilde{\d}_1&=\begin{bmatrix}
\text{real}(\d_1)\\
\text{imag}(\d_1)
\end{bmatrix}\\
\tilde{\d}_2&=\begin{bmatrix}
\text{real}(\d_2)\\
\text{imag}(\d_2)
\end{bmatrix}
\end{align}
\end{subequations}

Replace $\tilde{\P}\tilde{\d}$ in (\ref{eq:Pd_x_def1}) with $\tilde{\P}\tilde{\d}_1+\tilde{\Q}\tilde{\d}_2$ and redefine $\llambda$ and $\U$ the eigenvalues and unit eigenvectors of matrix $\tilde{\P}^T\tilde{\P}+\tilde{\Q}^T\tilde{\Q}$, the expressions of $\aalpha$ and $\bbeta$ can be rewritten as
\begin{subequations}\label{eq:Pd_x_SLL_def}
\begin{align}
\tilde{\P}\tilde{\d}_1+\tilde{\Q}\tilde{\d}_2&=\U\bbeta \label{eq:Pd_x_SLL_def1}\\
\tilde{\x}&=\U\aalpha \label{eq:Pd_x_SLL_def2}
\end{align}
\end{subequations}

Replace $\tilde{\P}^T\tilde{\P}$ with $\tilde{\P}^T\tilde{\P}+\tilde{\Q}^T\tilde{\Q}$,
problem $\min_{\x}f_1(\x)$ can be solved with the same strategy used for problem (\ref{eq:PGPS_without_SLL_x}).
$\aalpha$ in (\ref{eq:Pd_x_SLL_def2}) can be obtained utilizing the procedures in (\ref{eq:allpha}), (\ref{eq:nu_value}) and the algorithm in Appendix \ref{app:a}.
Hence, the algorithm proposed for solving problem (\ref{eq:PGPS_with_SLL_Lag}) can be summarized in \textbf{Algorithm \ref{alg:PGM-WSC}}.
For simplicity, mark the proposed algorithm as 'PGM-WSC', which is the abbreviation of 'Power Gain Maximization With SLL Constraint'.
The algorithm is terminated either when the predetermined maximum iteration is reached or when the predetermined estimation accuracies of $\g$ and $\h$ are reached.

\begin{algorithm}[ht]
\caption{The PGM-WSC algorithm for problem (\ref{eq:PGPS_with_SLL_Lag})}\label{alg:PGM-WSC}
\begin{algorithmic}[1]
\STATE initialize $\rho_1,\rho_2\in(1, 10000)$, $\x^{(0)}=\textbf{0}_{\text{L}_\text{ML}\times1}$, $\u_1^{(0)}=\textbf{0}_{\text{L}_\text{ML}\times1}$, $\u_2^{(0)}=\textbf{0}_{\text{L}_\text{ML}\times1}$
\WHILE{$ k<\textit{IterMax}$~$\text{or}$~$ \max\{|\P^H\x-\g|\}>10^{-4}$ and $\max\{|\Q^H\x-\h|\}>10^{-4}$}
\STATE update $\{g_0^{(k)}, \g^{(k)}, \h^{(k)}\}$ by minimizing $L_{\rho_1, \rho_2}(g_0, \g, \h)$ in  (\ref{eq:PGPS_with_SLL_gg1}), (\ref{eq:PGPS_with_SLL_gg2}) and (\ref{eq:PGPS_with_SLL_gg3}) with (\ref{eq:PGPS_with_SLL_gg_value1}), (\ref{eq:PGPS_with_SLL_gg_value2}) and (\ref{eq:PGPS_with_SLL_gg_value3}), respectively.\label{alg2:update_ggh}
\STATE update $\x^{(k)}$  by solving $\tilde{\x}$ in (\ref{eq:Pd_x_SLL_def2}) with the estimated $\aalpha$ via (\ref{eq:nu_value})\label{alg2:update_x}
\STATE update $\u_1^{(k)}$ with (\ref{eq:PGPS_with_SLL_u1})\label{alg2:update_u1}
\STATE update $\u_2^{(k)}$ with (\ref{eq:PGPS_with_SLL_u2})\label{alg2:update_u2}
\STATE $k=k+1$
\ENDWHILE
\STATE return array weight $\w=\C^{-1}\x$.
\end{algorithmic}
\end{algorithm}

\subsection{Computational Complexity}
Multiplications involved in one iteration of two proposed algorithms are compared with those of the PGPS-based algorithms in \cite{LeiYHZCQ_2019, LeiHCTPX_2020, ZhangJZ_2020e}:
\begin{itemize}
\item PGPS-based algorithms in \cite{LeiYHZCQ_2019, LeiHCTPX_2020, ZhangJZ_2020e}: They exhibit similar form as that used in \cite{LiangXZL_2017}, which has been proved to be computed with complexity $O\{\max\{N, \text{L}_\text{ML}+\text{L}_\text{SL}\}^4N^{0.5}\log(1/\epsilon)\}$ where $\epsilon\thickapprox10^{-8}$.
    Since algorithm in \cite{ZhangJZ_2020e} has a two-stage independent loop, its complexity is doubled in each single iteration.
\item PGM-WoSC in Algorithm \ref{alg:PGM-WoSC}: the computation mainly includes three parts: 1). the computation of $\{g_0, \g\}$ in step \ref{alg1:update_gg} is $O\{N\text{L}_\text{ML}\}$; 2). the computation of $\x$ in step \ref{alg1:update_x} is $O\{N^2\text{L}_\text{ML}\log N\}$; 3). the computation of $\u$ in step \ref{alg1:update_x} is $O\{N\text{L}_\text{ML}\}$, Hence, the total computation is $O\{N\text{L}_\text{ML}+N^2\text{L}_\text{ML}\log N+N\text{L}_\text{ML}\}=O\{N^2\text{L}_\text{ML}\log N\}$.
\item PGM-WSC in Algorithm \ref{alg:PGM-WSC}: the computation mainly includes four parts: 1). the computation of $\{g_0, \g, \h\}$ in step \ref{alg2:update_ggh} is $O\{N\max\{\text{L}_\text{ML},\text{L}_\text{SL}\}\}$; 2). the computation of $\x$ in step \ref{alg2:update_x} is $O\{N^2\max\{\text{L}_\text{ML}, \text{L}_\text{SL}\}\log N\}$; 3). the computation of $\u_1$ in step \ref{alg2:update_u1} is $O\{N\text{L}_\text{ML}\}$;
    and 4). the computation of $\u_2$ in step \ref{alg2:update_u2} is $O\{N\text{L}_\text{SL}\}$.
     Hence, the total computation is $O\{N\max\{\text{L}_\text{ML},\text{L}_\text{SL}\}+N^2\max\{\text{L}_\text{ML}, \text{L}_\text{SL}\}\log N+N\text{L}_\text{ML}+N\text{L}_\text{SL}\}=O\{N^2\max\{\text{L}_\text{ML}, \text{L}_\text{SL}\}\log N\}$.
\end{itemize}

The analysis shows that the proposed PGM-WoSC has a similar computational complexity compared with that of the proposed PGM-WSC, and both of them are much faster than the PGPS-based algorithms presented in \cite{LeiYHZCQ_2019, LeiHCTPX_2020, ZhangJZ_2020e}.
Besides,  \cite{LeiHCTPX_2020} is the fastest among \cite{LeiYHZCQ_2019, LeiHCTPX_2020, ZhangJZ_2020e} as it does not require iterative process while \cite{ZhangJZ_2020e} is the slowest due to the utilization of the two-stage iteration processes.
\begin{figure*}[!t]
\centering
\subfigure[]{\includegraphics[width=2.4in]{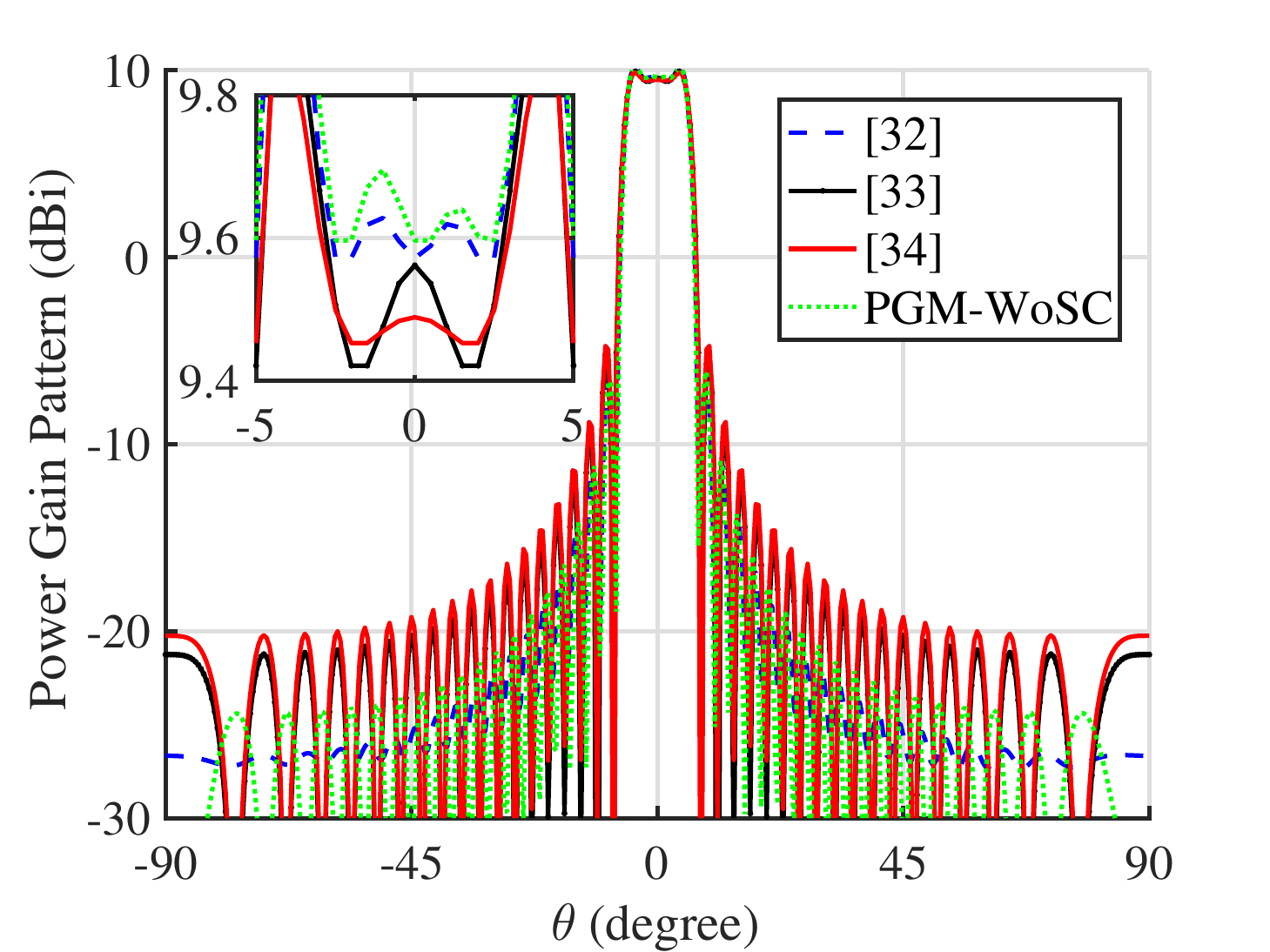}\label{fig:IEP_wd10}}
\subfigure[]{\includegraphics[width=2.4in]{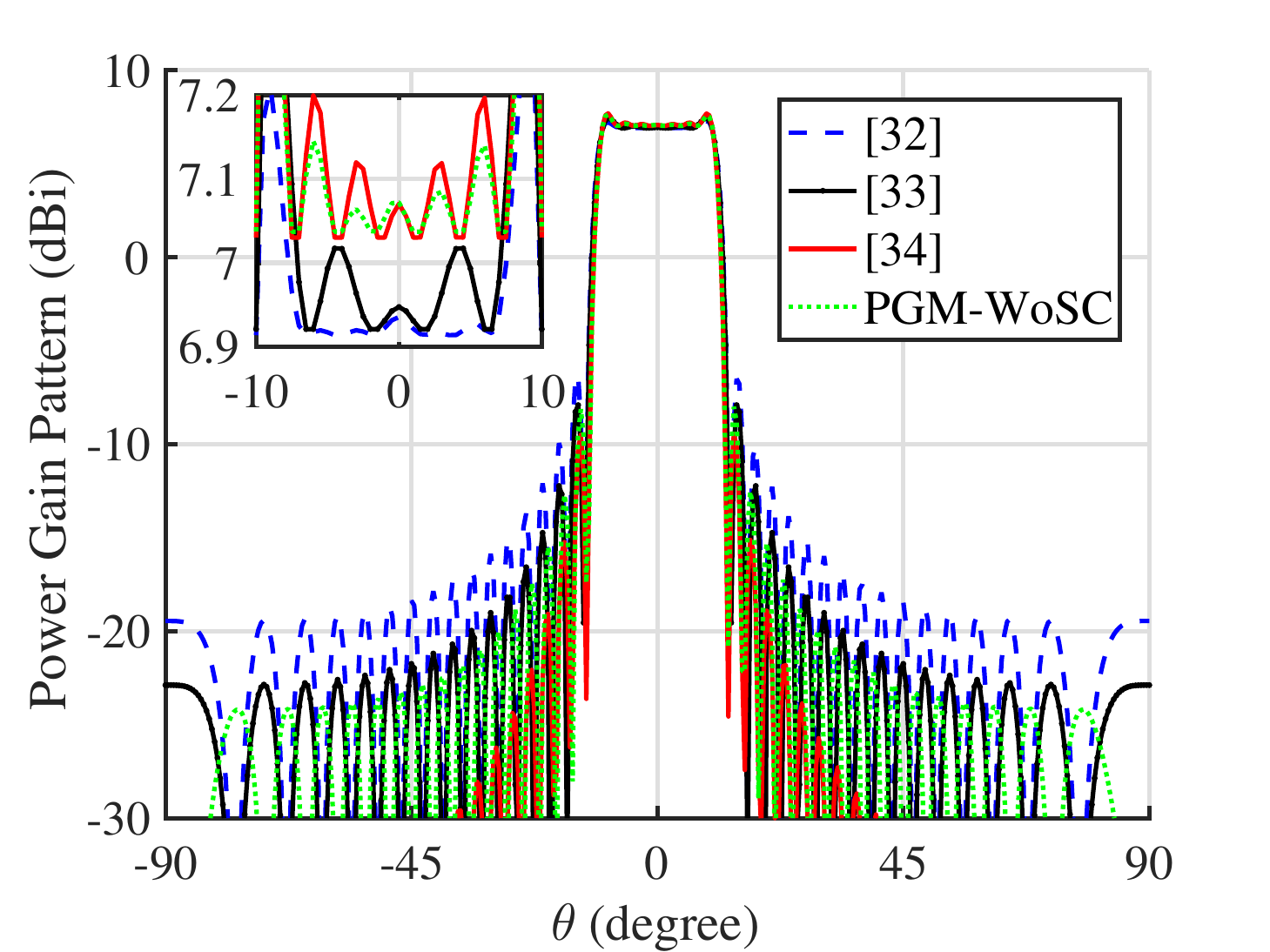}\label{fig:IEP_wd20}}
\subfigure[]{\includegraphics[width=2.4in]{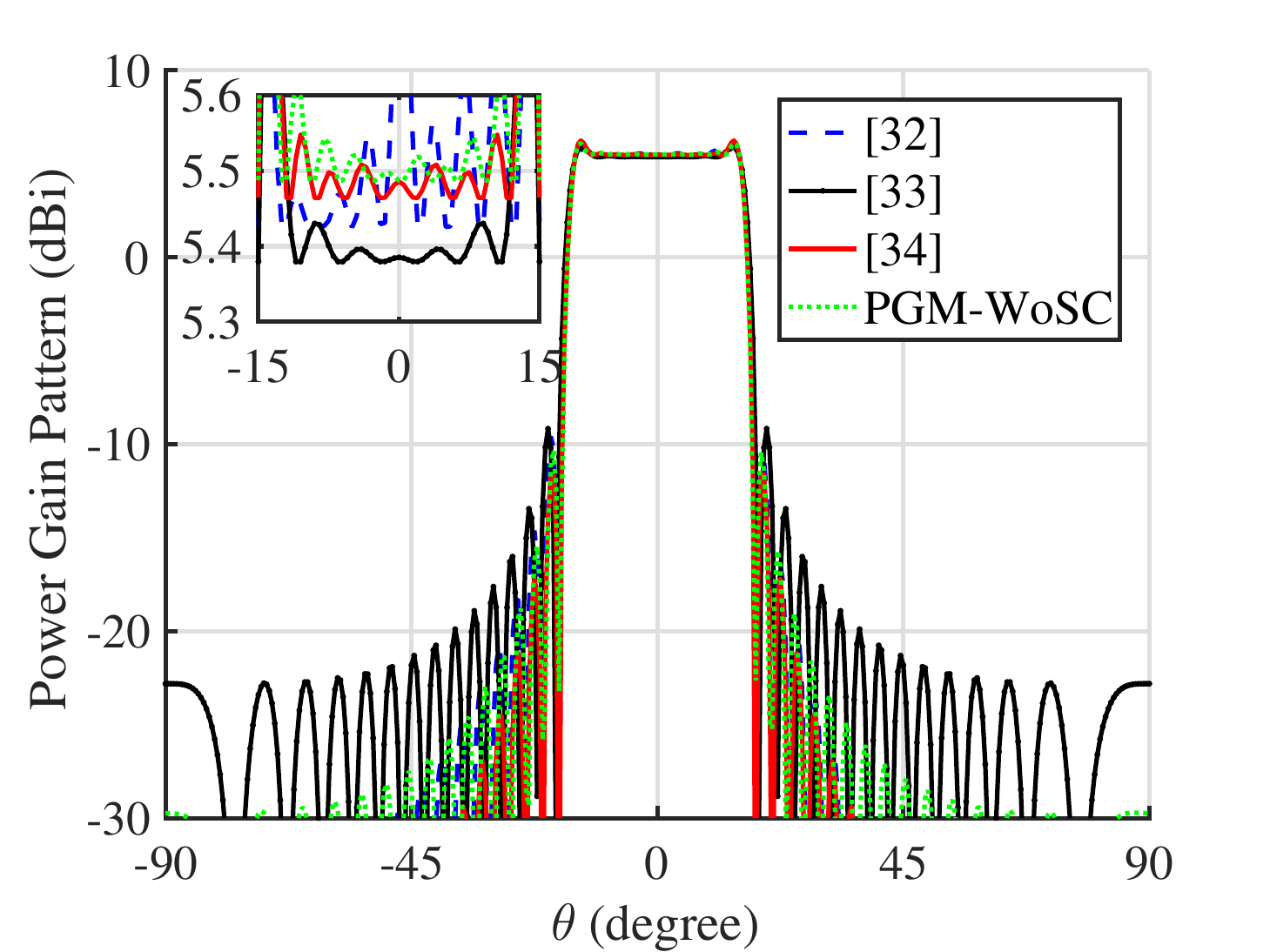}\label{fig:IEP_wd30}}
\subfigure[]{\includegraphics[width=2.4in]{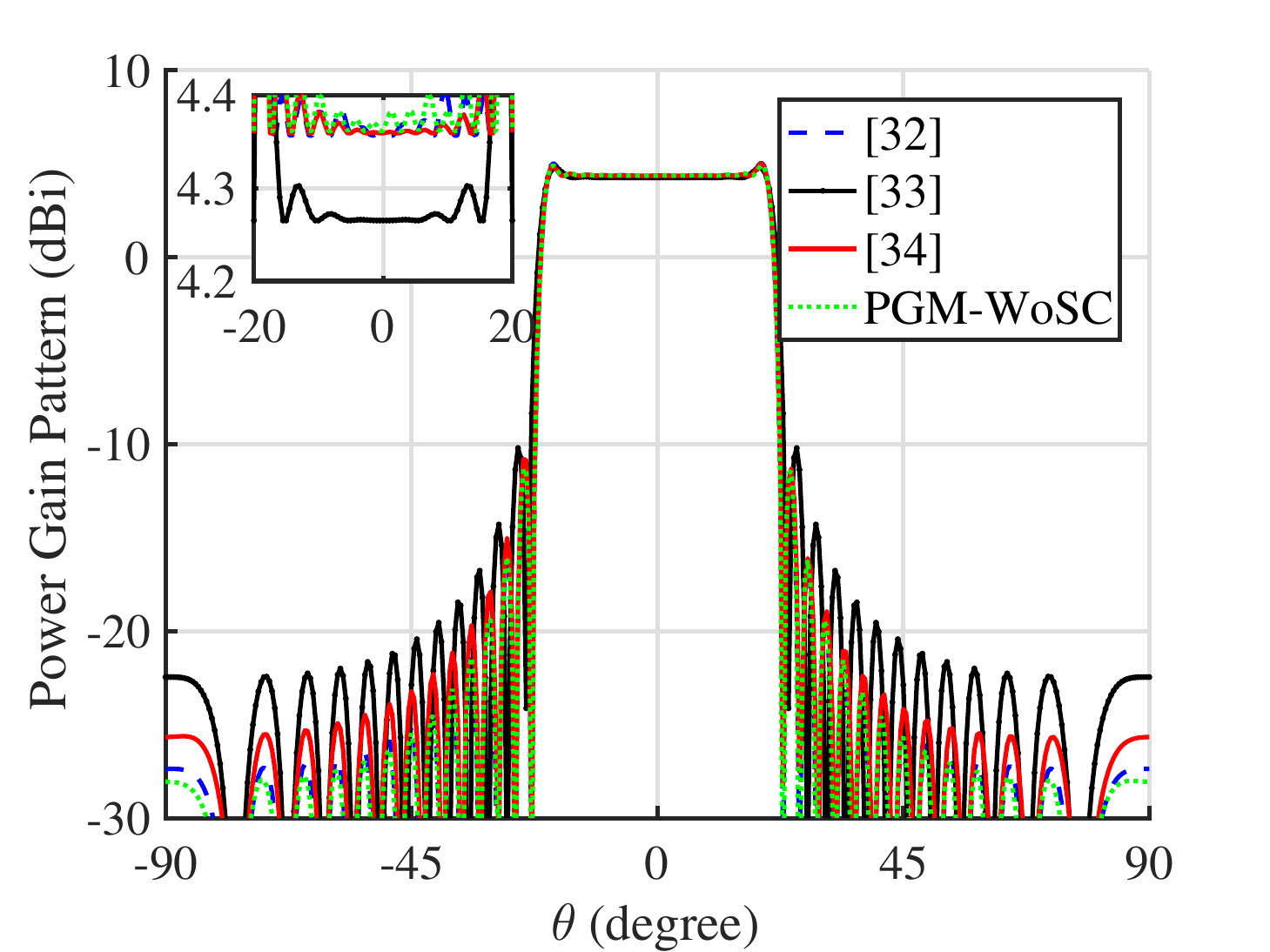}\label{fig:IEP_wd40}}
\caption{The synthesized power gain patterns with different beamwidth via $41$-element linear array antenna with the IEP and $\theta_c=0^o$: (a) $\boldsymbol{\Theta}_\text{bw}=10^o$, (b) $\boldsymbol{\Theta}_\text{bw}=20^o$, (c) $\boldsymbol{\Theta}_\text{bw}=30^o$, and (d) $\boldsymbol{\Theta}_\text{bw}=40^o$.}
\label{fig:IEP_wd}
\end{figure*}
\subsection{Convergence Analysis}
The convergence of the classic ADMM algorithm to handle convex optimization problem has been discussed in \cite{BoydPCPE_2011}.
However, in this paper the convergence of the ADMM with nonconvex problem needs to be further discussed \cite{WenYLM_2012}.
Here, a proposition is introduced to prove that the proposed algorithms can converge to a local optimal solution under some mild assumptions.
\begin{proposition} \label{pro:proposition1}
Consider the update steps in \text{(\ref{eq:PGPS_with_SLL_iter})}, if the following assumptions hold,
\begin{subequations}\label{eq:Convergence}
\begin{align}\label{eq:Convergence_1}
\lim \limits_{{k} \to \infty} (\u_1^{(k+1)}-\u_1^{(k)})=\textbf{0}
\end{align}
\begin{align}\label{eq:Convergence_2}
&\lim \limits_{{k} \to \infty} (\u_2^{(k+1)}-\u_2^{(k)})=\textbf{0}
\end{align}
\begin{align}\label{eq:Convergence_3}
&\rho_1>0,~\rho_2>0
\end{align}
\end{subequations}
then, the sequence $\{\x^{(k+1)}$, $g_0^{(k+1)}$, $\g^{(k+1)}$, $\h^{(k+1)}$, $\u_1^{(k+1)}$, $\u_2^{(k+1)}\}$ can converge to the locally optimal solution $\{\x^{(*)}, g_0^{(*)}, \g^{(*)}, \h^{(*)}, \u_1^{(*)}, \u_2^{(*)}\}$.
\end{proposition}

\begin{proof}
According to (\ref{eq:PGPS_with_SLL_u1})-(\ref{eq:PGPS_with_SLL_u2}) and (\ref{eq:Convergence_1})-(\ref{eq:Convergence_2}), we can obtain

\begin{subequations}\label{eq:Prof_c}
\begin{align}\label{eq:Prof_c_1}
\lim \limits_{{k} \to \infty} (\P^{H}\x^{(k+1)}-\g^{(k+1)})=\textbf{0}
\end{align}
\begin{align}\label{eq:Prof_c_2}
\lim \limits_{{k} \to \infty} (\Q^{H}\x^{(k+1)}-\h^{(k+1)})=\textbf{0}
\end{align}
\end{subequations}

Considering the following inequality:
\begin{subequations}\label{eq:Prof_d}
\begin{align}\label{eq:Prof_d_1}
\begin{split}
||\g^{(k+1)}||_2 &\leq ||\P^{H}\x^{(k+1)}-\g^{(k+1)}||_2+||\P^{H}\x^{(k+1)}||_2 \\
&\leq ||\P^{H}\x^{(k+1)}-\g^{(k+1)}||_2+||\P^{H}||_F||\x^{(k+1)}||_2
\end{split}
\end{align}
\begin{align}\label{eq:Prof_d_1}
\begin{split}
||\h^{(k+1)}||_2 &\leq ||\Q^{H}\x^{(k+1)}-\h^{(k+1)}||_2+||\Q^{H}\x^{(k+1)}||_2 \\
&\leq ||\Q^{H}\x^{(k+1)}-\h^{(k+1)}||_2+||\Q^{H}||_F||\x^{(k+1)}||_2
\end{split}
\end{align}
\end{subequations}

As $||\Q^{H}||_F$ and $||\P^{H}||_F$ are constants and $||\x||_2=1$ is bounded and closed set, the sequence $\{\x^{(k)}\}$ is bounded.
Hence, there exists a local point $\x^{(*)}$ where
\begin{align}\label{eq:Prof_x}
\lim \limits_{{k} \to \infty} \x^{(k+1)}=\x^{(*)}
\end{align}

Then, the sequence $\{\g^{(k)}\}$ and $\{\h^{(k)}\}$ can converge to a limit point $\{\g^{(*)}\}$ and $\{\h^{(*)}\}$.
In addition, according to inequality $g_0\leq g_l$, the sequence $\{g_0^{(k)}\}$ is also bounded.
Hence, there must exist such a stationary point $\{g_0^{(*)}\}$ that $\lim \limits_{{k} \to \infty}  g_0^{(k)} =g_0^{(*)}$.
This completes the proof.
\end{proof}

According to the proposition \ref{pro:proposition1}, if the dual variables, i.e. $\u_1$ and $\u_2$, gradually converge to a stable value, the proposed PGM-WSC algorithm can obtain its local optimal solution for the formulated PGPS problem with the SLL constraint.
The convergence of the proposed PGS-WoSC algorithm can be proved in a similar way.

\section{Simulations and Discussions}
This section carries out several numerical simulations to validate the superiority of the proposed algorithms by comparing the simulation results with those achieved from existing algorithms.

\subsection{Power Gain Optimization for Wide-Beam Without Sidelobe Constraint}
When the power gain in the wide-beam mainlobe is the main concerned parameter, the PGPS-based algorithms without SLL constraint  has been proved to be superior to the existing SBPS-based algorithms  in our previous works\cite{LeiYHZCQ_2019, LeiHCTPX_2020}.
In this part, the proposed PGM-WoSC is only compared with the PGPS-based algorithms in \cite{LeiYHZCQ_2019, LeiHCTPX_2020, ZhangJZ_2020e} by maximizing the minimum power gain in the mainlobe without any sidelobe restriction.

A $41$-element origin-symmetric half-wavelength uniformly distributed linear array antenna is simulated.
The radiation pattern of the element is assumed to be isotropic element pattern (IEP).
The parameters setting are as follows: the mainlobe beamwidth is $\boldsymbol{\Theta}_\text{bw}=10^o$, $20^o$, $30^o$ and $40^o$, the angular resolution $\Delta\theta=0.5^o$ and the angular distance between $\boldsymbol{\Theta}_\text{ML}$ and $\boldsymbol{\Theta}_\text{SL}$ is $3^o$.
The parameters set for different algorithms are as follows:
\begin{itemize}
\item \cite{LeiYHZCQ_2019}: $\alpha=0.2$ and $\beta=0.01$, the maximum iteration is $200$ and $gap=10^{-2}$.
\item \cite{LeiHCTPX_2020}: no more parameter settings.
\item \cite{ZhangJZ_2020e}: The initial guess of $\x$ is obtained via pseudo-inverse excitation (PIE) scheme, $c_0=1$, $c_1=0.9$, $\delta_{\max}=0.015$, the maximum iteration is $Itr_{\max}=100$ and $gap_0=10^{-4}$.
\item The proposed PGM-WoSC: $\rho = 1000$ with a decreasing factor of $0.99$ for each iteration and the maximum iteration is $2000$.
\end{itemize}
\begin{figure}[!t]
\centering
\subfigure[]{\includegraphics[width=2.4in]{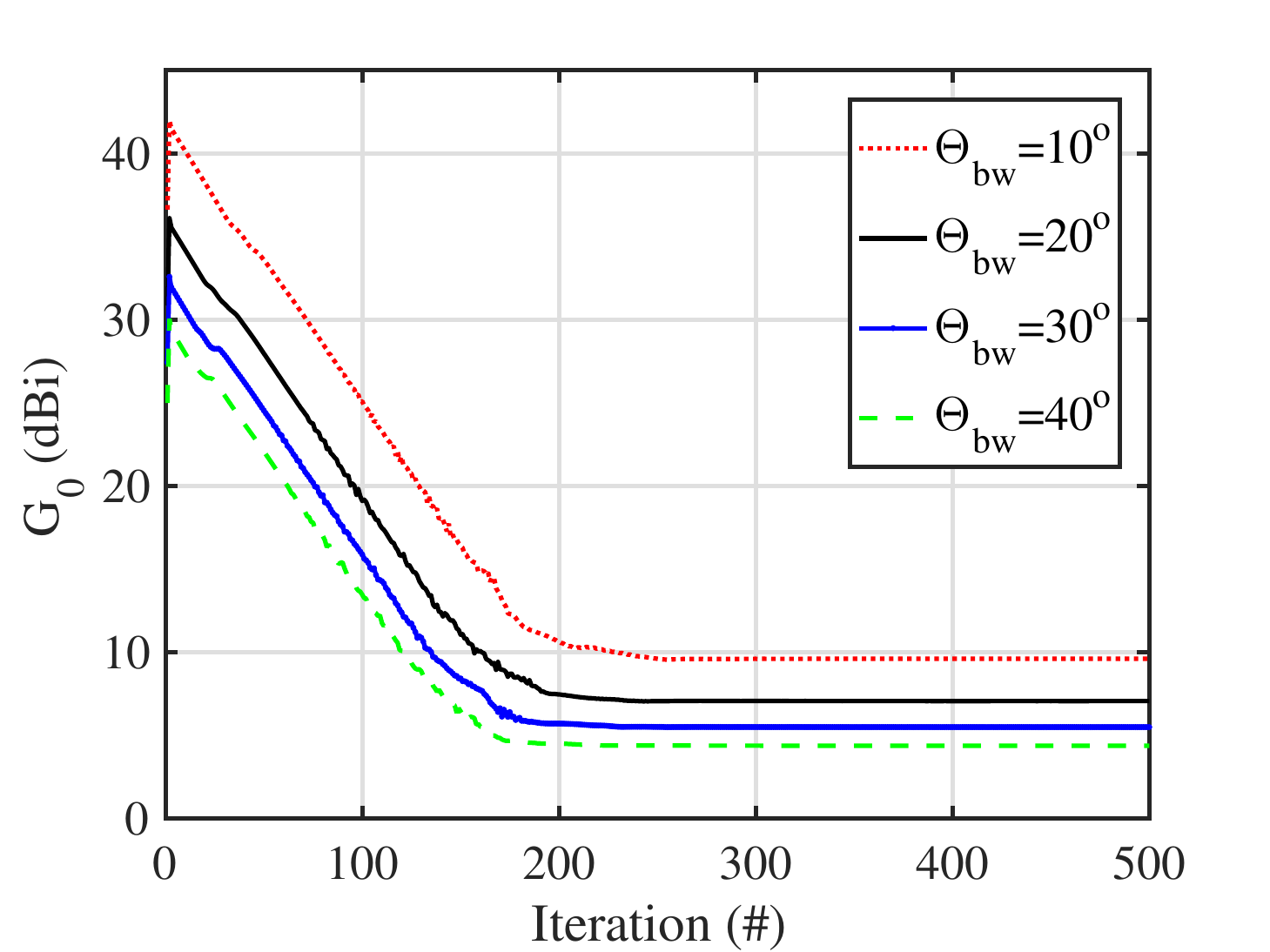}\label{fig:IEP_WoSC_iter_g0}}
\subfigure[]{\includegraphics[width=2.4in]{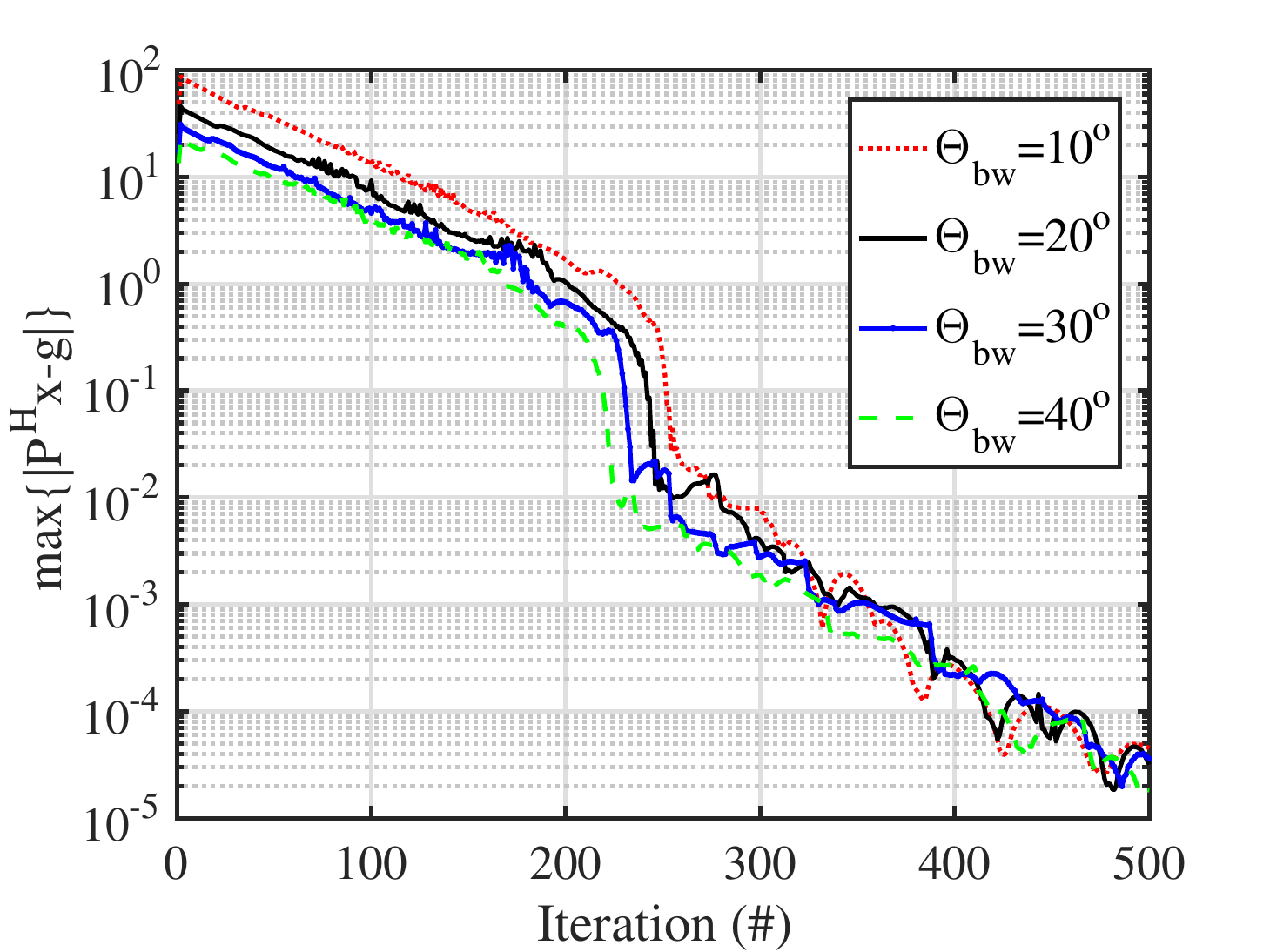}\label{fig:IEP_WoSC_iter_est}}
\caption{The convergence speed of $G_0$ and $\max\{|\P^H\x-\g|\}$ for the proposed PGM-WoSC algorithm w.r.t. the iteration number: (a) $G_0$ w.r.t the iteration number and (b) $\max\{|\P^H\x-\g|\}$ w.r.t the iteration number.}
\label{fig:IEP_WoSC_iter}
\end{figure}
\begin{figure}[!t]
\centering
\subfigure[]{\includegraphics[width=2.4in]{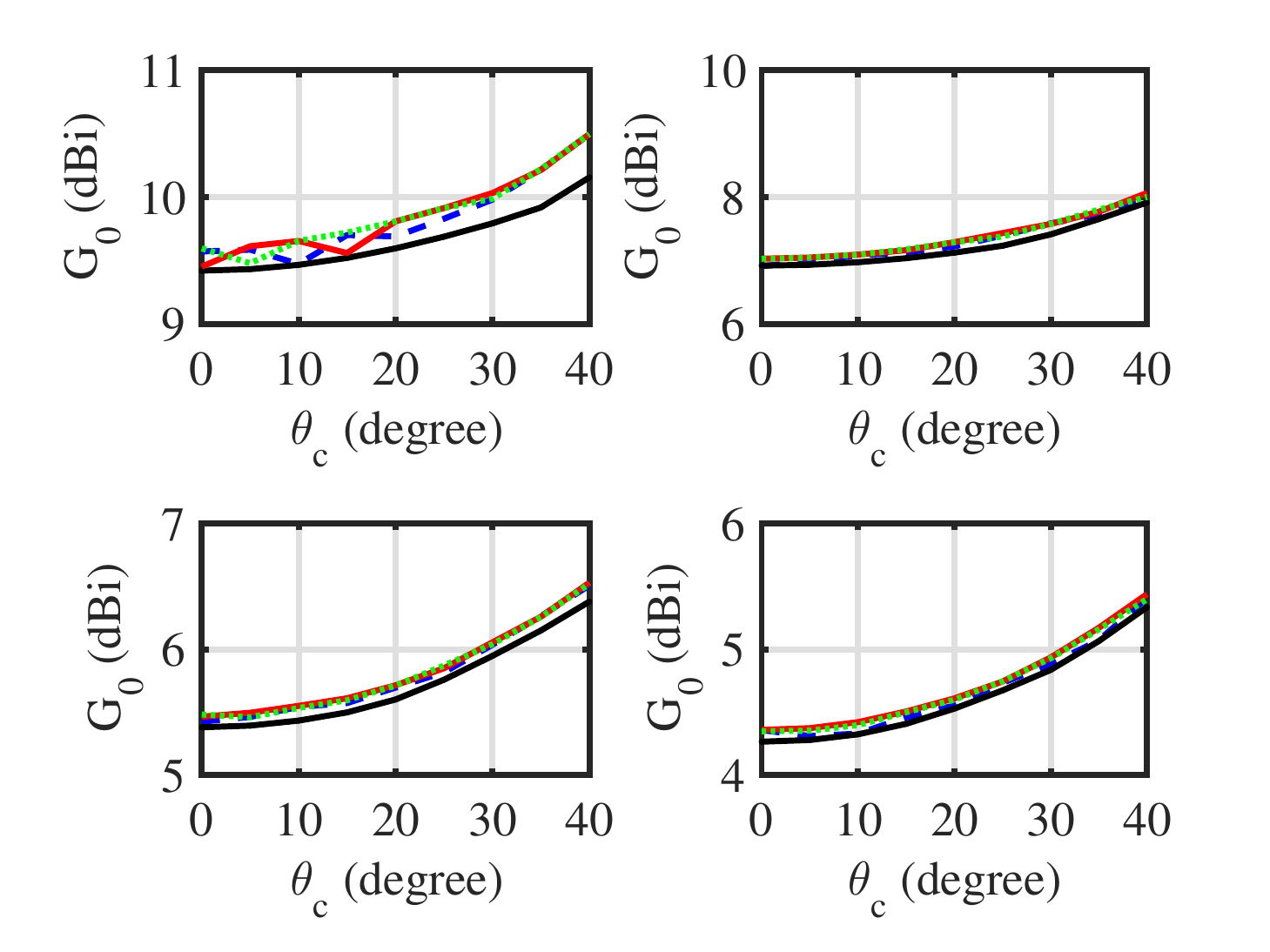}\label{fig:IEP_WoSC_g0}}
\subfigure[]{\includegraphics[width=2.4in]{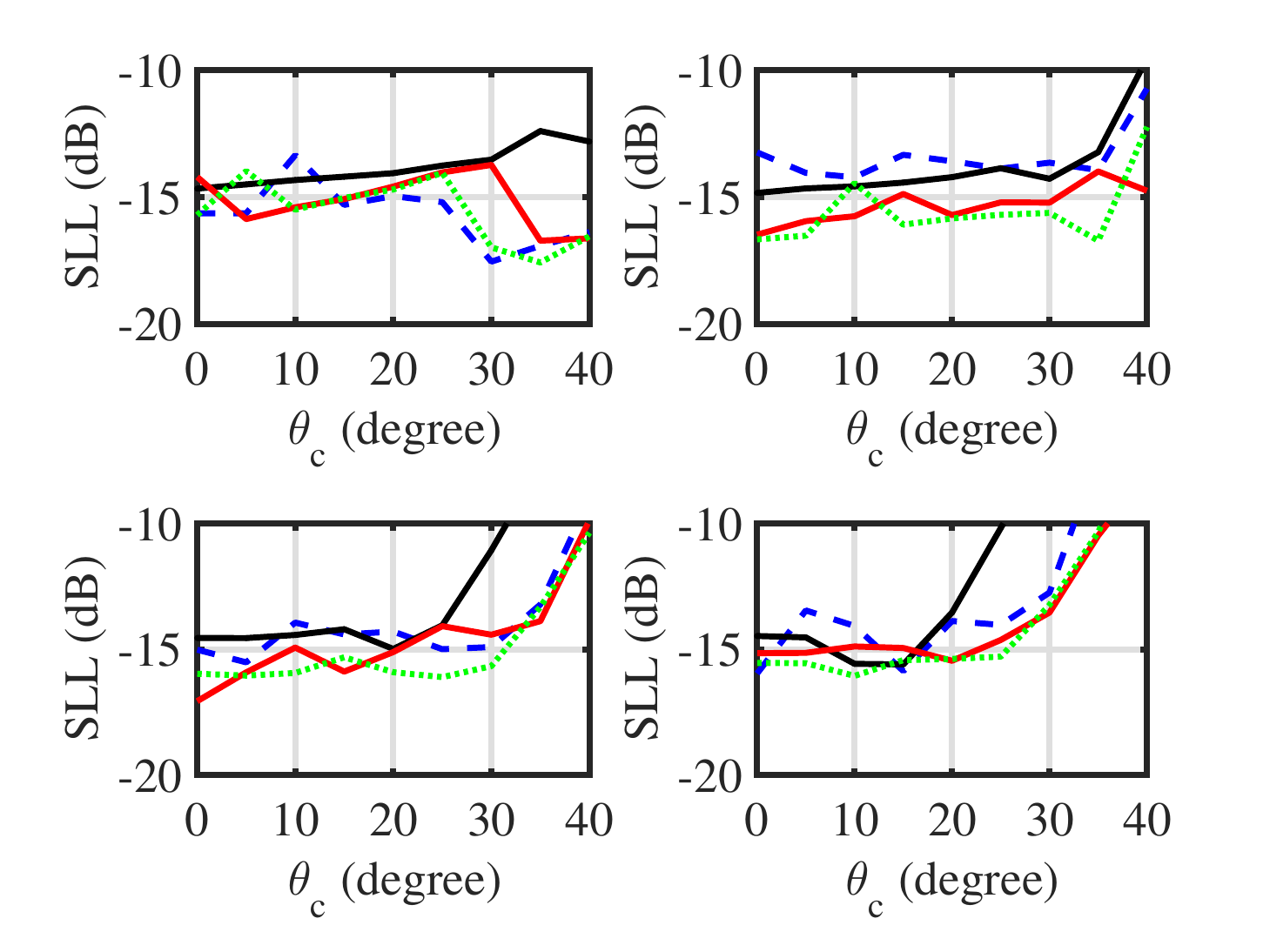}\label{fig:IEP_WoSC_PSLL}}
\caption{The synthesized wide-beam performance via uniformly distributed IEP array: (a) the max-minimum power gain in the wide-beam mainlobe  and (b) the peak SLL.
The top-left, top-right, down-left and down-right pictures in the subfigures represent the results with the beamwidth being $10^o, 20^0, 30^o$ and $40^o$, respectively.
The blue double dot dash line, the black dash dot line, the red solid line and the green dotted line represent the algorithms in \cite{LeiYHZCQ_2019, LeiHCTPX_2020, ZhangJZ_2020e} and the proposed PGM-WoSC, respectively.}
\label{fig:IEP_WoSC}
\end{figure}

\subsubsection{Example 1 (Maximize the Minimum Power Gain Without SLL Constraint via the IEP Array)}
\begin{table}[t]
\caption{The Max-Minimum Power Gain $G_0$ in the mainlobe without the SLL constraint with $\boldsymbol{\Theta}_\text{bw}=10^0, 20^o, 30^o$ and $40^o$ (unit: \textit{dBi})}\label{tab:tab_g0_WoSC}
\begin{center}
\begin{tabular}{c| c|  c| c| c}
\hline\hline
{\diagbox{Alg.}{$\boldsymbol{\Theta}_\text{bw}$}} & $10^o$ &$20^o$ &$30^o$ &$40^o$ \\
\hline\hline
\cite{LeiYHZCQ_2019} &9.57   &$6.91$ &$5.42$ &$4.35$\\
\hline
\cite{LeiHCTPX_2020} &9.42  &$6.92$  &$5.38$ &$4.26$ \\
\hline
\cite{ZhangJZ_2020e} &9.45  &$7.03$  &$5.46$ &$4.36$ \\
\hline
PGM-WoSC &9.59&$7.04$ &$5.49$  &$4.36$    \\
\hline\hline
\end{tabular}
\end{center}
\end{table}
Fig. \ref{fig:IEP_wd} plots the power gain patterns achieved with four algorithms under different beamwidth $\boldsymbol{\Theta}_\text{bw}$.
The results demonstrate that all algorithms can obtain wide-beam mainlobe and the proposed algorithm can always obtain the highest max-minimum power gains for different $\boldsymbol{\Theta}_\text{bw}$.
The concrete $G_0$ values are provided in Table \ref{tab:tab_g0_WoSC}, which further shows that the proposed PGM-WoSC algorithm performs better than the other algorithms.
Fig. \ref{fig:IEP_WoSC_iter} shows the convergence speed of $G_0$ and $\max\{|\P^H\x-\g|\}$ w.r.t. the iteration number. The results indicate that $G_0$ reaches its optimal value after $200$ iterations while $\max\{|\P^H\x-\g|\}$ is smaller than $10^{-4}$ after $500$ iterations.
It, therefore, proves that the proposed PGM-WoSC algorithm can converge to its suboptimal solution quite fast since the procedures in each iteration is pseudo-analytically.
\begin{figure*}[!t]
\centering
\subfigure[]{\includegraphics[width=2.4in]{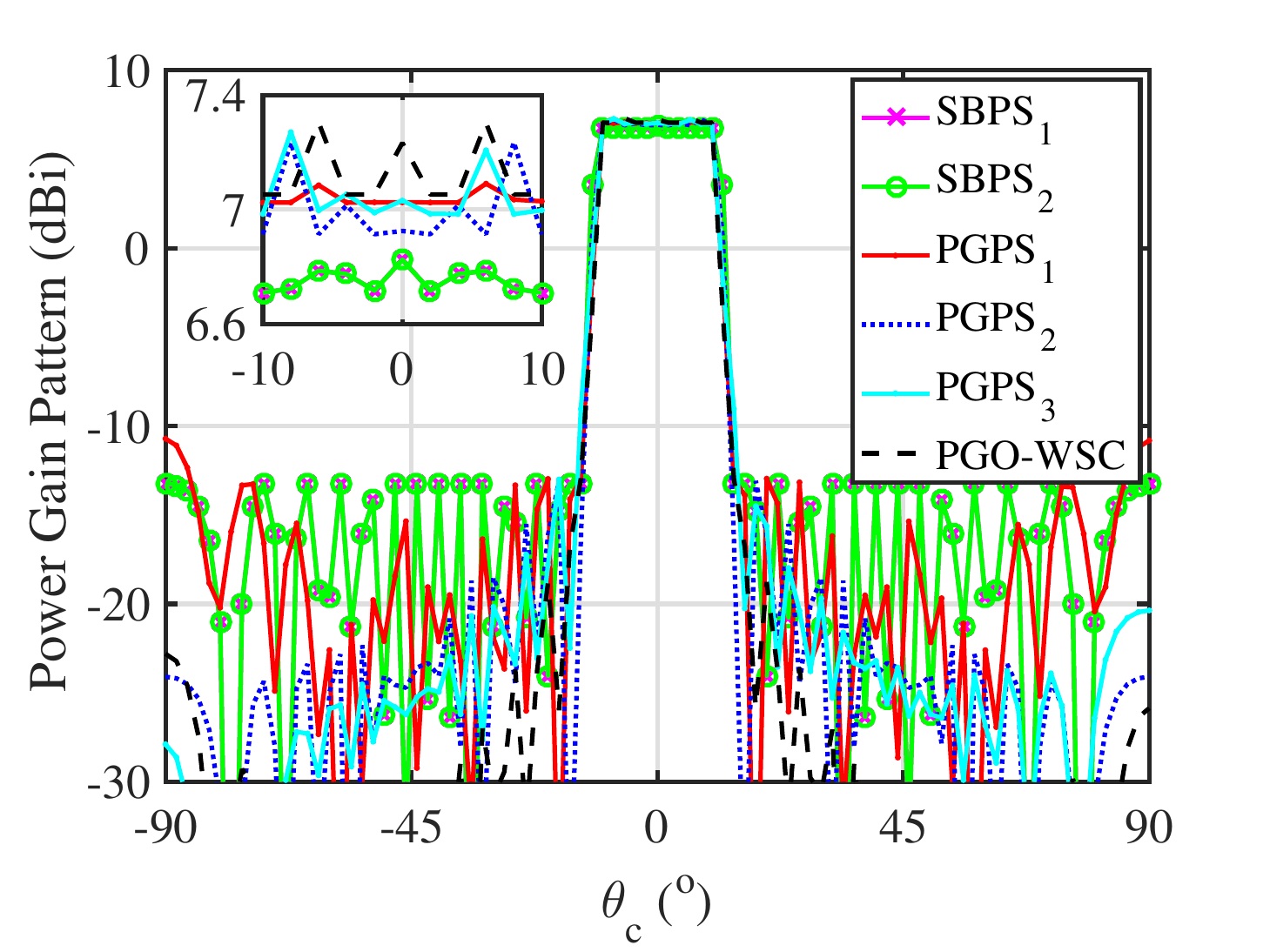}\label{fig:IEP_SLL20_wd20}}
\subfigure[]{\includegraphics[width=2.4in]{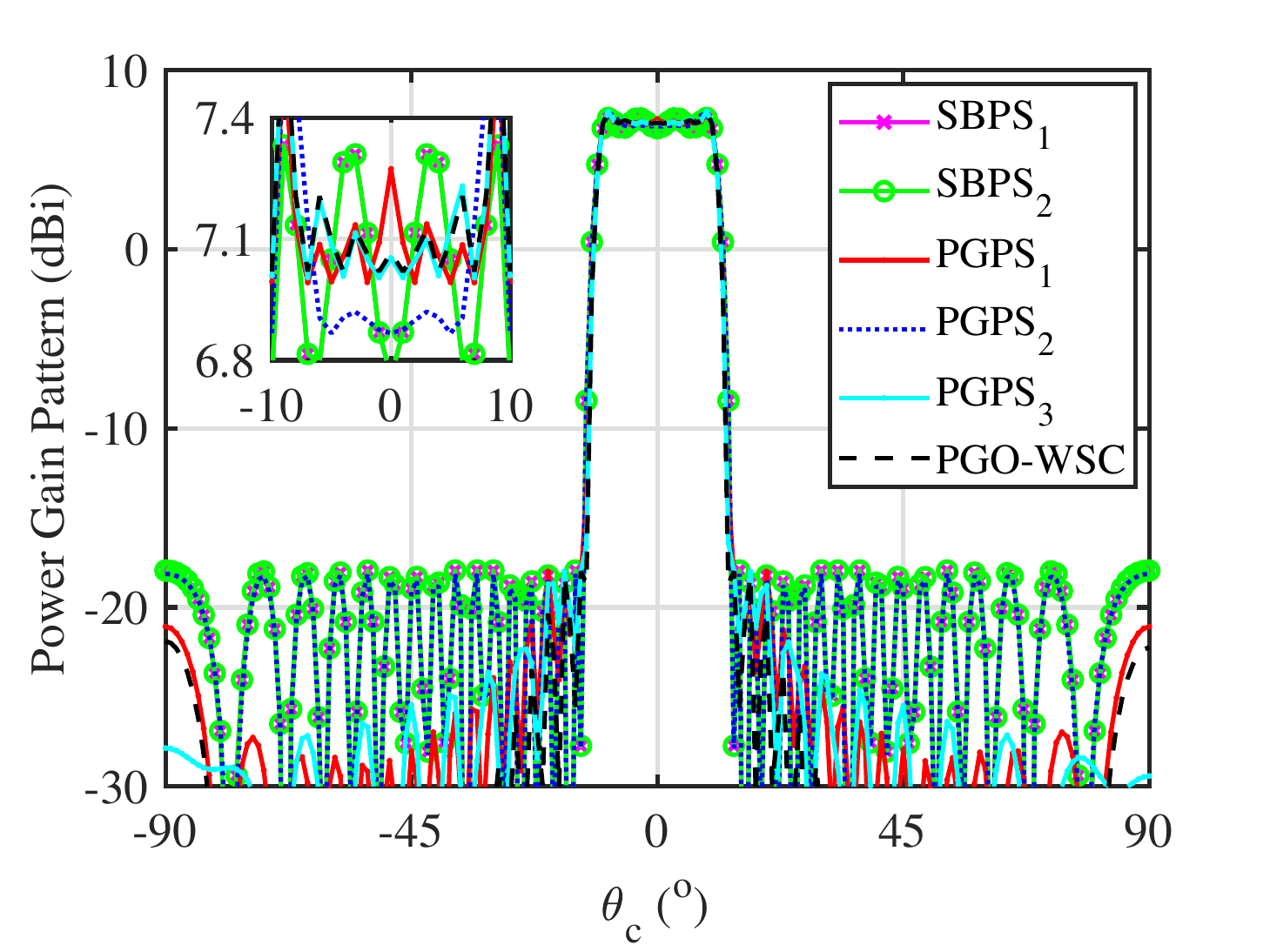}\label{fig:IEP_SLL25_wd20}}
\subfigure[]{\includegraphics[width=2.4in]{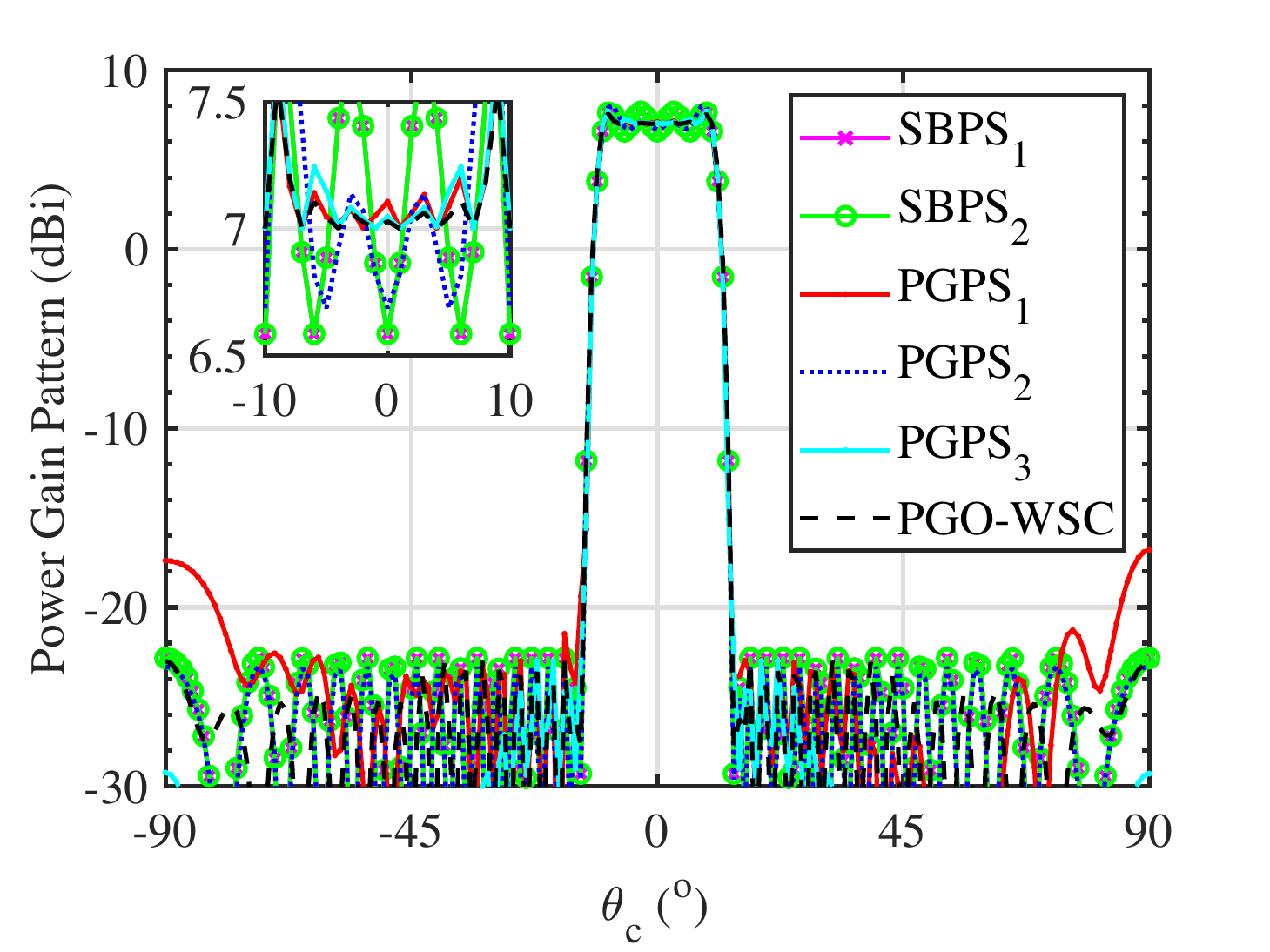}\label{fig:IEP_SLL30_wd20}}
\subfigure[]{\includegraphics[width=2.4in]{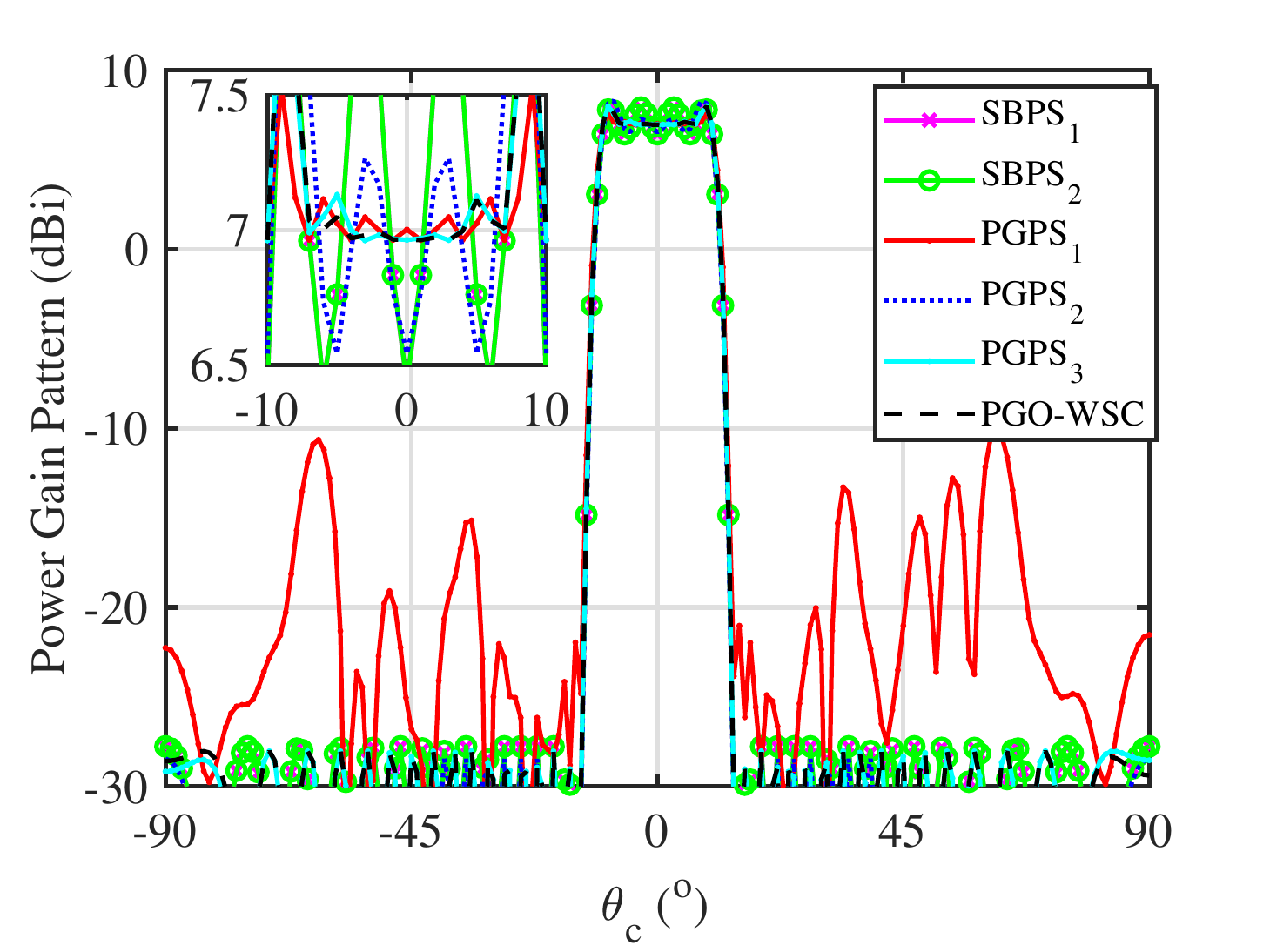}\label{fig:IEP_SLL35_wd20}}
\caption{The synthesized power gain patterns with different  dSLL (i.e. $\gamma$) constraints via $41$-element non-uniformly distributed linear IEP array antenna and $\theta_c=0^o$, $\boldsymbol{\Theta}_\text{bw}=20^o$: (a) $\gamma=-20$ dB, (b) $\gamma=-25$ dB, (c) $\gamma=-30$ dB, and (d) $\gamma=-35$ dB.}
\label{fig:IEP_SLL_wd}
\end{figure*}
\subsubsection{Example 2 (Scanning Without SLL Constraint via the IEP Array)}
To further validate the scanning ability, Fig. \ref{fig:IEP_WoSC_g0} shows the max-minimum power gain $G_0$ in the wide-beam mainlobe and Fig. \ref{fig:IEP_WoSC_PSLL} shows the peak SLL in the sidelobe region, when the wide-beam center $\theta_c$ scans from $0^o$ to $40^o$.
The top-left, top-right, down-left and down-right pictures in each subfigure represent the results for the beamwidth  $\boldsymbol{\Theta}_\text{bw}=10^o$, $20^o$, $30^o$ and $40^o$, respectively.
Statically, the proposed algorithm can obtain a better performance.
More precisely, the PGM-WoSC obtains higher or at least similar max-minimum power gain compared with the existing algorithms with a lower peak SLL in most directions.
Some other cases with different beamwidth and directions are also simulated and the results are similar as those from the shown figures and tables.

Generally, these two examples and other unshown examples demonstrate the effectiveness of the proposed PGM-WoSC algorithm in dealing with the minimum power gain maximization problem without the sidelobe constraint.
The results validate that the proposed algorithm performs better than the existing algorithms with faster computation speed.
Regrettably, the SLL obtained via these algorithms are uncontrollable while the algorithms including the proposed PGM-WSC and other existing algorithms can effectively control the SLL. This is discussed in the following part.

\subsection{Power Gain Optimization for Wide-Beam With Sidelobe Constraint}
To test the SLL control performance, the proposed PGM-WSC algorithm is compared with several existing algorithms:
\begin{itemize}
\item SBPS$_1$: the SBPS-based algorithm in \cite{Fuchs_2012} that tries to minimize the mainlobe ripple when a desired SLL is selected through solving following optimization problem:
    \begin{align}
    \begin{split}
    &~~~~\min_\w ~\epsilon\\
    &s.t.~~ |f_{\w}(\theta)-f_d(\theta)|^2\leq\epsilon,~\theta\in\boldsymbol{\Theta}_\text{ML}\\
    &~~~~~~|f_d(\theta)-f_{\w}(\theta)|^2\leq\rho,~\theta\in\boldsymbol{\Theta}_\text{SL}
    \end{split}
    \end{align}
    where $f_{\w}(\theta)$ is the beam pattern to be synthesized, $f_d(\theta)$ is the desired beam pattern. For instance, $f_d(\boldsymbol{\Theta}_\text{ML})=1$ and $f_d(\boldsymbol{\Theta}_\text{SL})=0$ for wide-beam mainlobe synthesis.
\item SBPS$_2$: the SBPS-based algorithm in \cite{FuchsSM_2013}. This algorithm tries to minimize the mainlobe ripple with an iterative strategy by optimizing the beam pattern: the maximum iteration is $Itr_{\max}=50$.
\item PGPS$_1$: The PGPS-based algorithm in \cite{LeiYHZCQ_2019} that tries to maximize the minimum power gain in mainlobe with a desired SLL via an iterative: $\alpha=0.2$ and $\beta=0.01$, the maximum iteration is $50$ and $gap=0.01$.approach
\item PGPS$_2$: The PGPS-based algorithm in \cite{LeiHCTPX_2020} that tries to maximize the minimum power gain in the wide-beam mainlobe with a desired SLL via an linear programming approach.
\item PGPS$_3$: The PGPS-based algorithm in \cite{ZhangJZ_2020e} that tries to maximize the minimum power gain in the wide-beam mainlobe via a two-stage iterative approach: the initial guess of $\x$ is obtained via PIE scheme with $c_0=1$, $c_1=0.9$, $c_2=0.9$, $\delta_{\max}=0.015$, the maximum iteration is $Itr_{\max}=100$ and $gap_0=10^{-4}$.
\end{itemize}
\begin{table*}
\caption{The max-minimum power gain ($G_0$) in different beamwidth ($\boldsymbol{\Theta}_\text{bw}$ and different desired SLL (dSLL, $\gamma$), Unit: \textit{dBi}.}\label{tab:tab_SLL_G0}
\begin{center}
\begin{tabular}{l| c c c c| c c c c| c c c c| c c c c}
\hline\hline
{$\boldsymbol{\Theta}_\text{bw}$}&\multicolumn{4}{c|}{$10^o$}
&\multicolumn{4}{c|}{$20^o$} &\multicolumn{4}{c|}{$30^o$} &\multicolumn{4}{c}{$40^o$}\\ \hline
{\diagbox{Alg.}{$\text{dSLL}$}} &$-20$ &$-25$  &$-30$ &$-35$  &$-20$ &$-25$  &$-30$ &$-35$ &$-20$ &$-25$  &$-30$ &$-35$ &$-20$ &$-25$  &$-30$ &$-35$ \\
\hline\hline
SBPS$_1$ &9.14 &9.32 &9.03 &8.70 &6.78 &6.76 &6.59 &6.43 &5.27 &5.19 &5.09 &4.62 &4.14 &4.10 &3.79 &3.41\\
SBPS$_2$ &9.14 &9.32 &9.03 &8.71 &6.78 &6.76 &6.59 &6.43 &5.27 &5.19 &5.09 &4.62 &4.14 &4.10 &3.79 &3.42\\
PGPS$_1$ &9.57 &9.50 &9.54 &9.50 &7.03 &7.00 &7.00 &7.00 &5.49 &5.48 &5.43 &5.43 &4.34 &4.34 &4.33 &4.32\\
PGPS$_2$ &9.41 &9.37 &8.97 &8.63 &6.91 &6.87 &6.69 &6.54 &5.37 &5.31 &5.24 &4.92 &4.26 &4.23 &4.12 &3.75\\
PGPS$_3$ &9.55 &9.42 &9.41 &9.36 &7.00 &6.97 &6.95 &6.94 &5.49 &5.48 &5.45 &5.44 &4.33 &4.34 &4.32 &4.29\\
PGM-WSC  &9.59 &9.41 &9.40 &9.29 &7.03 &7.01 &6.98 &6.93 &5.47 &5.45 &5.45 &5.36 &4.34 &4.33 &4.33 &4.19\\
\hline \hline
\end{tabular}
\end{center}
\end{table*}
\begin{table*}
\caption{The obtained SLL (oSLL) in different beamwidth ($\boldsymbol{\Theta}_\text{bw}$ and different desired SLL (dSLL, $\gamma$), Unit: \textit{dB}.}\label{tab:tab_SLL_oPSLL}
\begin{center}
\begin{tabular}{l| c c c c| c c c c| c c c c| c c c c}
\hline\hline
{$\boldsymbol{\Theta}_\text{bw}$}&\multicolumn{4}{c|}{$10^o$}
&\multicolumn{4}{c|}{$20^o$} &\multicolumn{4}{c|}{$30^o$} &\multicolumn{4}{c}{$40^o$}\\ \hline
{\diagbox{Alg.}{$\text{dSLL}$}} &$-20$ &$-25$  &$-30$ &$-35$  &$-20$ &$-25$  &$-30$ &$-35$ &$-20$ &$-25$  &$-30$ &$-35$ &$-20$ &$-25$  &$-30$ &$-35$ \\
\hline\hline
SBPS$_1$ &-19.9 &-24.8 &-29.3 &-33.9 &-19.9 &-24.7 &-29.4 &-34.2 &-19.9 &-24.7 &-29.5 &-34.0 &-19.8 &-24.7 &-29.4 &-34.0\\
SBPS$_2$ &-19.9 &-24.8 &-29.3 &-33.8 &-19.9 &-24.7 &-29.4 &-34.2 &-19.9 &-24.7 &-29.5 &-34.0 &-19.8 &-24.7 &-29.3 &-33.9\\
PGPS$_1$ &-20.0 &-24.5 &-21.5 &-17.4 &-20.0 &-23.9 &-18.7 &-25.1 &-20.0 &25.0 &-24.4 &-16.5 &-20.0 &24.8 &-18.9 &-22.0\\
PGPS$_2$ &-20.0 &-25.0 &-30.0 &-35.0 &-20.0 &-25.0 &-30.0 &-35.0 &-20.0 &-25.0 &-30.0 &-35.0 &-20.0 &-25.0 &-30.0 &-35.0\\
PGPS$_3$ &-20.0 &-25.0 &-30.0 &-35.0 &-20.0 &-25.0 &-30.0 &-35.0 &-20.0 &-25.0 &-30.0 &-35.0 &-20.0 &-25.0 &-30.0 &-35.0\\
PGM-WSC  &-20.0 &-25.0 &-30.0 &-35.0 &-20.0 &-25.0 &-30.0 &-35.0 &-20.0 &-25.0 &-30.0 &-35.0 &-20.0 &-25.0 &-30.0 &-35.0\\
\hline \hline
\end{tabular}
\end{center}
\end{table*}

This part would validate the superiority of the proposed PGM-WSC algorithm with both the IEP and active element pattern (AEP) via a $41$-element origin-symmetric half-wavelength non-uniformly distributed linear array antenna.
The positive array element positions are (normalized with $\lambda$): $[0.6215,$ $1.0414,$ $1.4743,$ $1.9572,$ $2.5043,$ $2.9870,$ $3.4492,$ $4.0155,$ $4.4617,$ $4.9544,$ $ 5.3895,$ $5.8762,$ $6.3107,$ $6.8955,$ $7.4536,$ $7.9957,$ $8.4196,$ $8.9315,$ $ 9.4274,$ $10 ]\lambda$.
Set the mainlobe beamwidth as $\boldsymbol{\Theta}_\text{bw}=10^o$, $20^o$, $30^0$ and $40^o$, the mainlobe region and the sidelobe region are defined accordingly by setting their nearest angular distance as $3^o$.  $\rho_1=\rho_2=2000$ with a decreasing factor $0.99$ is used for each iteration. The desired and obtained SLL are marked as dSLL and oSLL, respectively.

\subsubsection{Example 3 (Maximize the Minimum Power Gain With SLL Constraint via the IEP Array)}
Fig. \ref{fig:IEP_SLL_wd} plots the power gain pattern obtained via different algorithms by setting the dSLL $\gamma=-20$ dB, $-25$ dB, $-30$ dB and $-35$ dB, respectively.
The results show that the proposed PGM-WSC algorithm always performs the best.
For the latter four algorithms, they obtain similar max-minimum power gain in the wide-beam mainlobe.
However, the proposed algorithm can always obtain the dSLL, i.e., oSLL$=$dSLL, while the PGPS$_1$ cannot when $\gamma=-30$ dB and $\gamma=-35$ dB.
In addition, the proposed algorithm obtains about $0.2$ dB to $0.3$ dB lower peak SLL than the PGPS$_2$.
Though the PGPS$_3$ can obtain similar max-minimum power gain to the proposed algorithm, it is about $10$ times slower ($150.5$ seconds versus $14.5$ seconds).
Besides, the PGPS-based algorithms, i.e., PGPS$_1$, PGPS$_2$ and PGM-WSC, are always able to obtain higher $G_0$ than the SBPS-based algorithms, i.e.,  SBPS$_1$ and SBPS$_2$.
Hence, the proposed PGM-WSC algorithm is more robust and faster compared with the existing algorithms in dealing with different wide-beam power gain synthesis problems.

Figure \ref{fig:IEP_WSC_iter_u} shows the convergence speed of $\max\{|\P^H\x-\g|\}$ and $\max\{|\Q^H\x-\h|\}$ w.r.t. the iteration number.
The results show that both parameters are smaller than $10^{-4}$ after $600$ iterations.
Hence, it proves that the PGM-WSC algorithm converges to its suboptimal solution within hundreds of iterations.
Since the problem-solving process in every single iteration is pseudo-analytical, the proposed algorithm would run faster.
\begin{table}
\caption{The Running Time of the PGPS$_1$, the PGPS$_3$ and the proposed PGM-WSC algorithm with $\boldsymbol{\Theta}_{\text{ML}}=20^o$ (unit: second).}\label{tab:tab_run_time}
\begin{center}
\begin{tabular}{c | c| c| c}
\hline\hline
{\diagbox{$\theta_c$}{Alg.}} &PGPS$_1$ &PGPS$_3$  &PGM-WSC   \\
\hline\hline
$0^o$ &18.4 &332.2 &13.8 \\
$5^o$ &16.8 &425.1 &14.5 \\
$10^o$ &16.9 &392.3 &13.8 \\
$15^o$ &17.1 &295.7 &16.2 \\
$20^o$ &16.8 &235.9 &14.3 \\
$25^o$ &17.2 &451.6 &15.7 \\
$30^o$ &17.3 &689.1 &14.4 \\
$35^o$ &17.0 &262.9 &14.8 \\
$40^o$ &17.1 &293.6 &14.6 \\
\hline \hline
\end{tabular}
\end{center}
\end{table}
Some other cases include $\boldsymbol{\Theta}_\text{bw} = \{10^o, 30^o, 40^o\}$ are also simulated.
The obtained $G_0$ and oSLL are provided in Table \ref{tab:tab_SLL_G0} and \ref{tab:tab_SLL_oPSLL}, respectively.
The oSLL is defined as the gap between the minimum value in the mainlobe region and the maximum value in the sidelobe region.
Table \ref{tab:tab_SLL_G0} shows that PGM-WSC, PGPS$_1$ and PGPS$_3$ always obtain a higher $G_0$ than the other three algorithms.
Meanwhile, PGM-WSC and PGPS$_3$ can always obtain the dSLL, i.e., dSLL$=$oSLL, while PGPS$_1$ cannot, especially when the beamwidth $\boldsymbol{\Theta}_\text{bw}$ is large and the dSLL $\gamma$ is small.
For instance, when $\boldsymbol{\Theta}_\text{bw}=30^o, 40^o$ and $\gamma=-30, -35$ dB (i.e., dSLL$=-30, -35$ dB), the oSLL obtained via the PGPS$_1$ is about $10$ dB higher.
The proposed algorithm exhibit similar performance as PGPS$_3$. But, as discussed in Fig. \ref{fig:IEP_SLL_wd}, it is much faster than  PGPS$_3$.
Notice that oSLL via the SBPS-based algorithms, i.e., SBPS$_1$ and SBPS$_2$, is different from that provided in the related references.
This is caused by the definition difference of the oSLL.
The oSLL is defined as the difference between the middle value (the ripple is considered to vary around the middle value) in the mainlobe region and the maximum value in the sidelobe region in existing references, while it is defined as the difference between the minimum value in the mainlobe region and the maximum value in the sidelobe region.

\subsubsection{Example 4 (Scanning with SLL Constraint via the IEP Array)}
\begin{figure}[!t]
\centering
\subfigure[]{\includegraphics[width=2.4in]{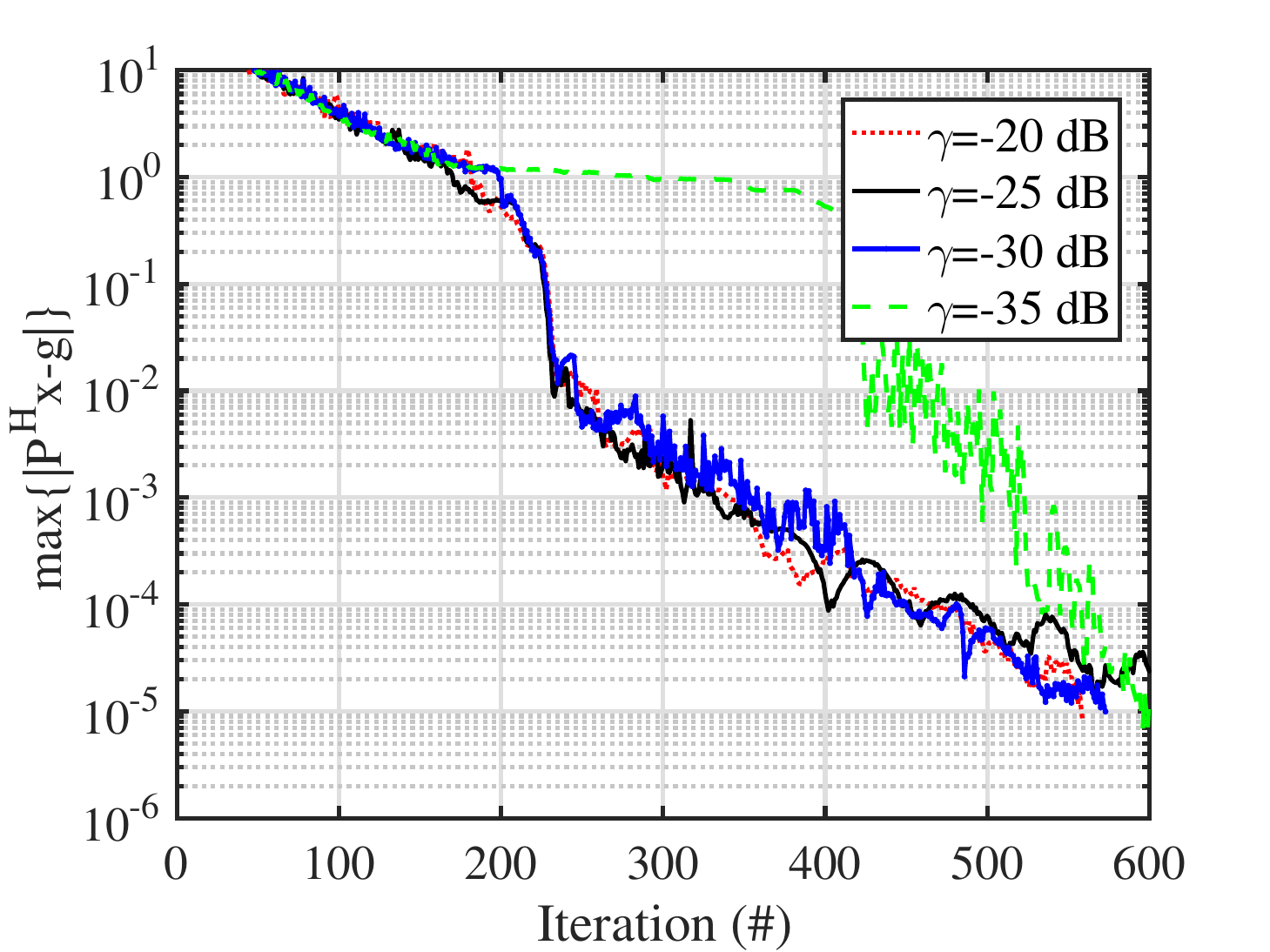}\label{fig:IEP_WSC_iter_u1}}
\subfigure[]{\includegraphics[width=2.4in]{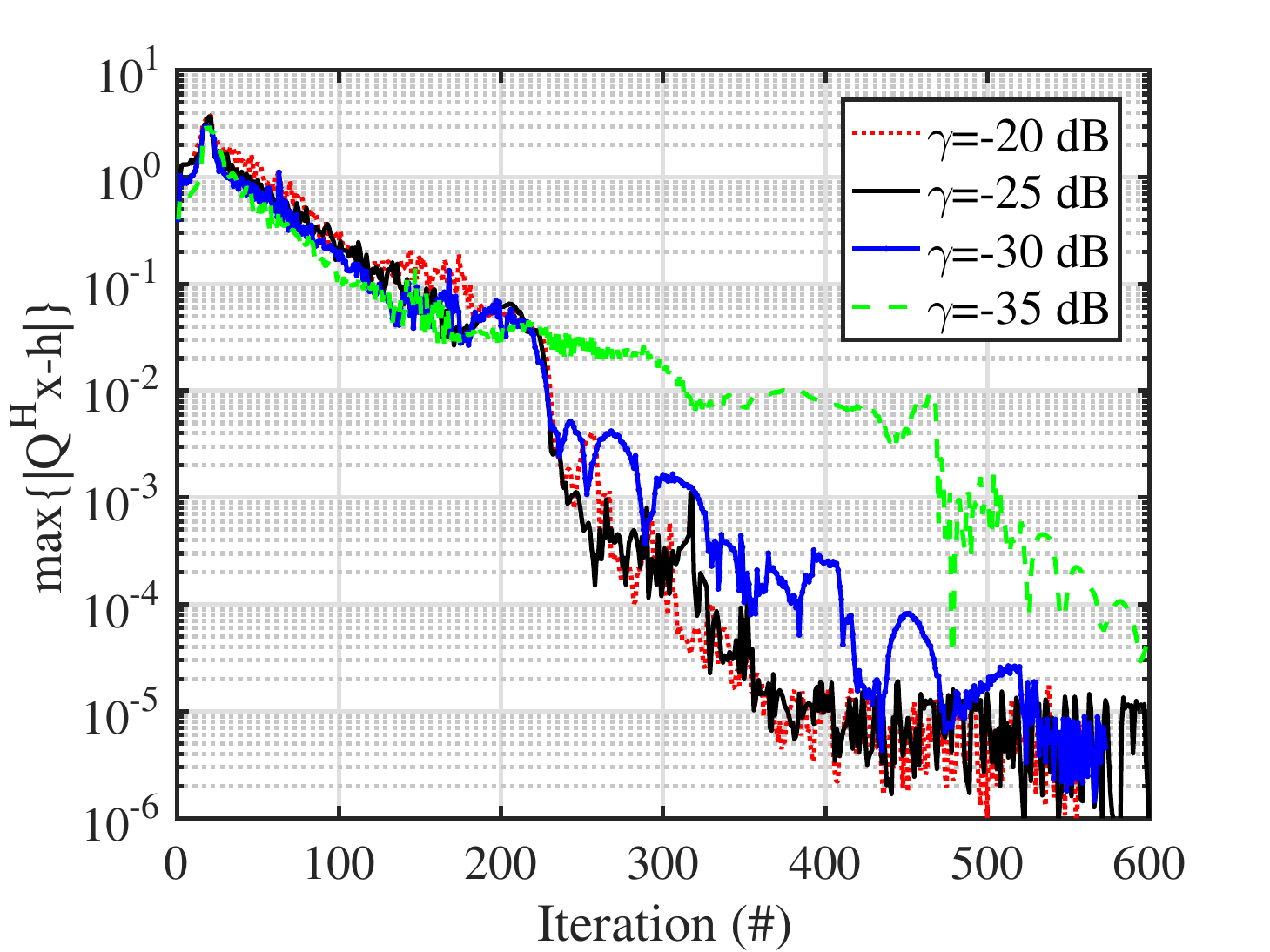}\label{fig:IEP_WSC_iter_u2}}
\caption{The convergence speed of $\max\{|\P^H\x-\g|\}$ and $\max\{|\Q^H\x-\h|\}$ for the proposed PGM-WSC algorithm w.r.t. the iteration number: (a) $\max\{|\P^H\x-\g|\}$ w.r.t the iteration number and (b) $\max\{|\Q^H\x-\h|\}$ w.r.t the iteration number.}
\label{fig:IEP_WSC_iter_u}
\end{figure}
\begin{figure}[!t]
\centering
\subfigure[]{\includegraphics[width=2.4in]{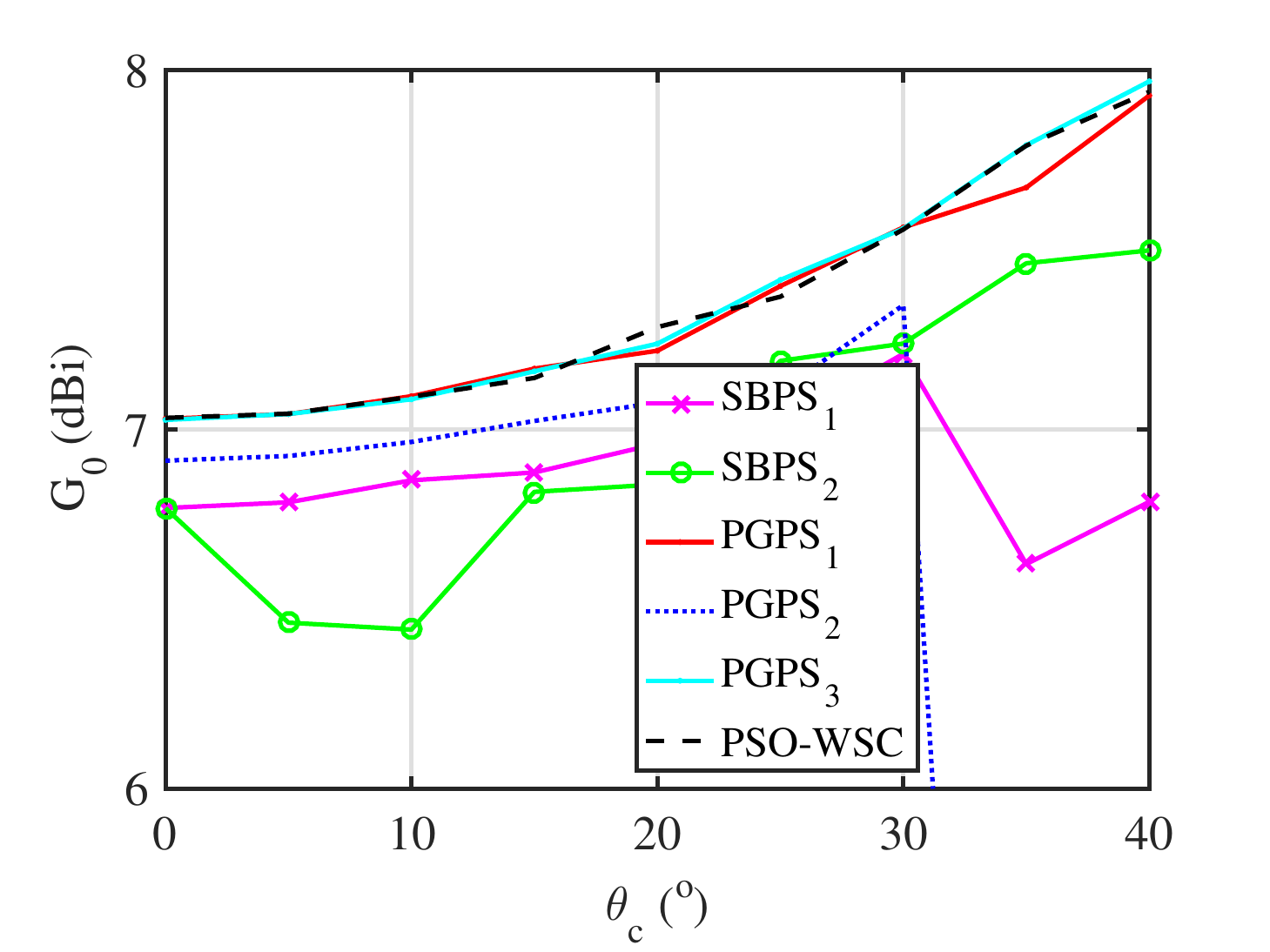}\label{fig:IEP_SLL20_wd10_Scanning_G0}}
\subfigure[]{\includegraphics[width=2.4in]{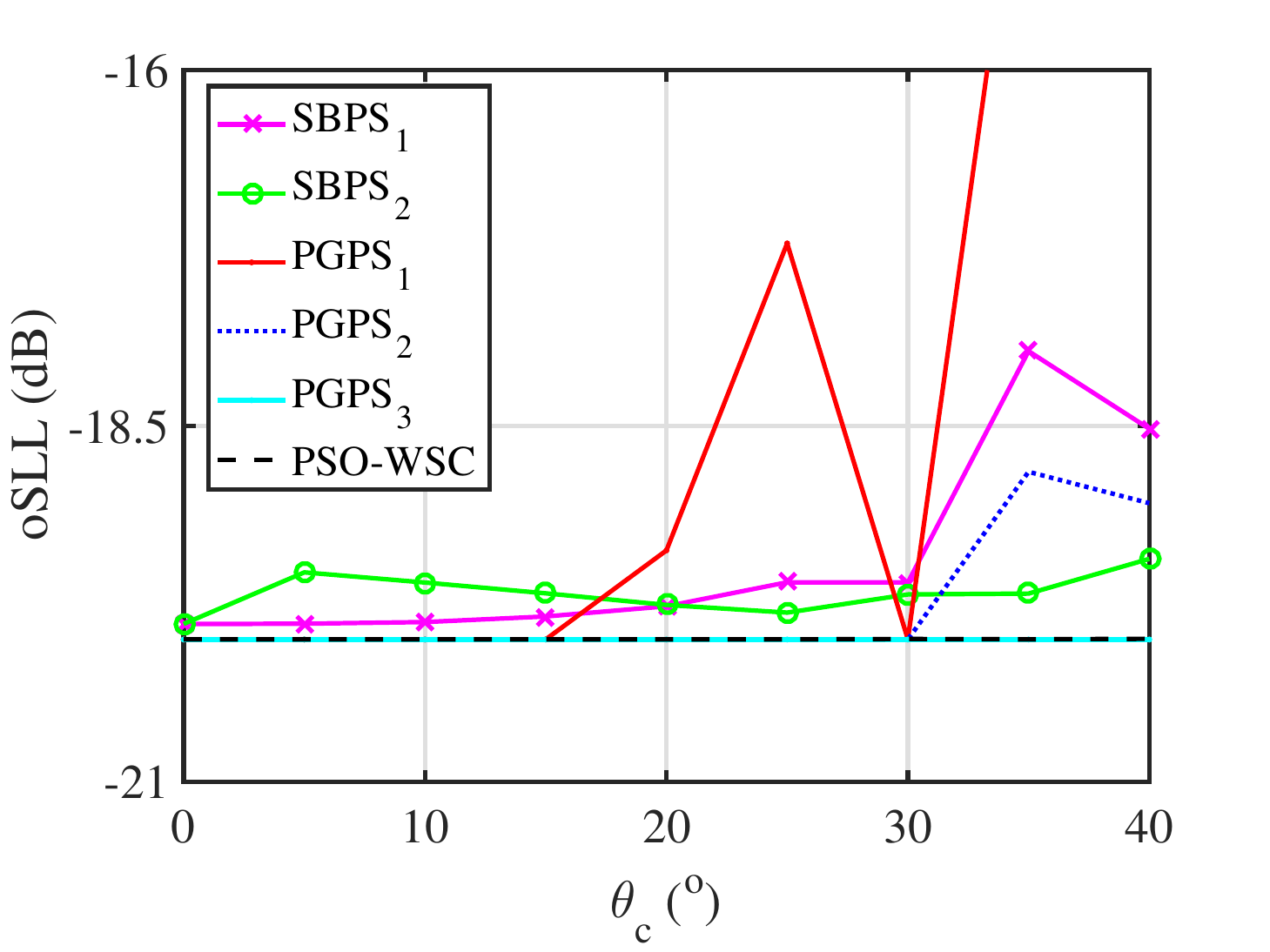}\label{fig:IEP_SLL20_wd10_Scanning_PSLL}}
\caption{The obtained $G_0$ and the obtained SLL for synthesizing a wide-beam via non-uniformly distributed IEP array with the dSLL$=-20$ dB and the beamwidth $\boldsymbol{\Theta}_\text{bw}=10^o$: (a)  the obtained $G_0$  and (b) the obtained SLL.}\label{fig:IEP_SLL20_wd10_Scanning}
\end{figure}

This part validates the scanning ability of the proposed PGM-WSC for synthesizing the wide-beam power gain pattern.
Fig. \ref{fig:IEP_SLL20_wd10_Scanning} plots the obtained max-minimum power gain, i.e., $G_0$ and the oSLL, by setting the dSLL$=-20$ dB and $\boldsymbol{\Theta}_\text{bw}=20^o$, when the beam center $\theta_c$ scans from $0^o$ to $40^o$.
It shows that the PGPS-based algorithms can always obtain a higher minimum mainlobe power gain $G_0$ than the SBPS-based algorithms.
PGPS$_1$,  PGPS$_3$ and the proposed algorithm have similar $G_0$s (which exhibit a fluctuation up to $\pm0.1$ dB for different beam directions) that are higher than $G_0$s achieved via other three algorithms.
Among these three algorithms, only PGM-WSC and PGPS$_3$ can obtain the desired SLL when the SLL is restricted (see Fig. \ref{fig:IEP_SLL20_wd10_Scanning_PSLL}).
In the sense of SLL control, PGM-WSC and PGPS$_3$ are better than PGPS$_1$.
Table \ref{tab:tab_run_time} shows the running times for the beam scanning.
It indicates that PGM-WSC and PGPS$_1$ are at least $10$ times faster than PGPS$_3$.
With all those comparisons, it can be concluded that PGM-WSC exhibits the best performance and a lower computational complexity.
In addition, other cases include dSLL$=\{-25$, $-30$, $-35\}$ dB,  $\boldsymbol{\Theta}_\text{bw}=\{30^o, 40^o\}$ are also simulated.
All results validate that the proposed PGM-WSC obtains the highest $G_0$ with the predetermined SLL and a lower computational burden.

\subsubsection{Example 5 (Performance Validation with AEP Array)}
\begin{table*}
\caption{The max-minimum power gain ($G_0$) and the obtained SLL (the oSLL) with different desired SLL (unit: \textit{dBi} for $G_0$ and \textit{dB} for oSLL).}\label{tab:tab_AEP_SLL_G0_PSLL}
\begin{center}
\begin{tabular}{l| c c c c| c c c c}
\hline\hline
{$\sharp$}&\multicolumn{4}{c|}{G$_0$}
&\multicolumn{4}{c}{oSLL}\\ \hline
{\diagbox{Alg.}{dSLL}} &$-20$ &$-25$  &$-30$ &$-35$  &$-20$ &$-25$  &$-30$ &$-35$   \\
\hline\hline
SBPS$_1$ &6.78 &6.79 &6.60 &6.43 &-13.1 &-18.0 &-22.9 &-27.8\\
SBPS$_2$ &6.78 &6.79 &6.60 &6.45 &-19.9 &-24.7 &-29.4 &-34.2\\
PGPS$_1$ &7.02 &7.02 &7.01 &7.01 &-20.0 &-25.0 &-26.3 &-25.3 \\
PGPS$_2$ &6.92 &6.90 &6.71 &6.54 &-20.0 &-25.0 &-30.0 &-35.0\\
PGPS$_3$ &7.03 &7.03 &7.01 &6.99 &-20.0 &-25.0 &-30.0 &-35.0\\
PGM-WSC  &7.03 &7.02 &7.01 &6.99 &-20.0 &-25.0 &-30.0 &-35.0 \\
\hline \hline
\end{tabular}
\end{center}
\end{table*}
\begin{figure}[!t]
\centering
\includegraphics[width=2.2in]{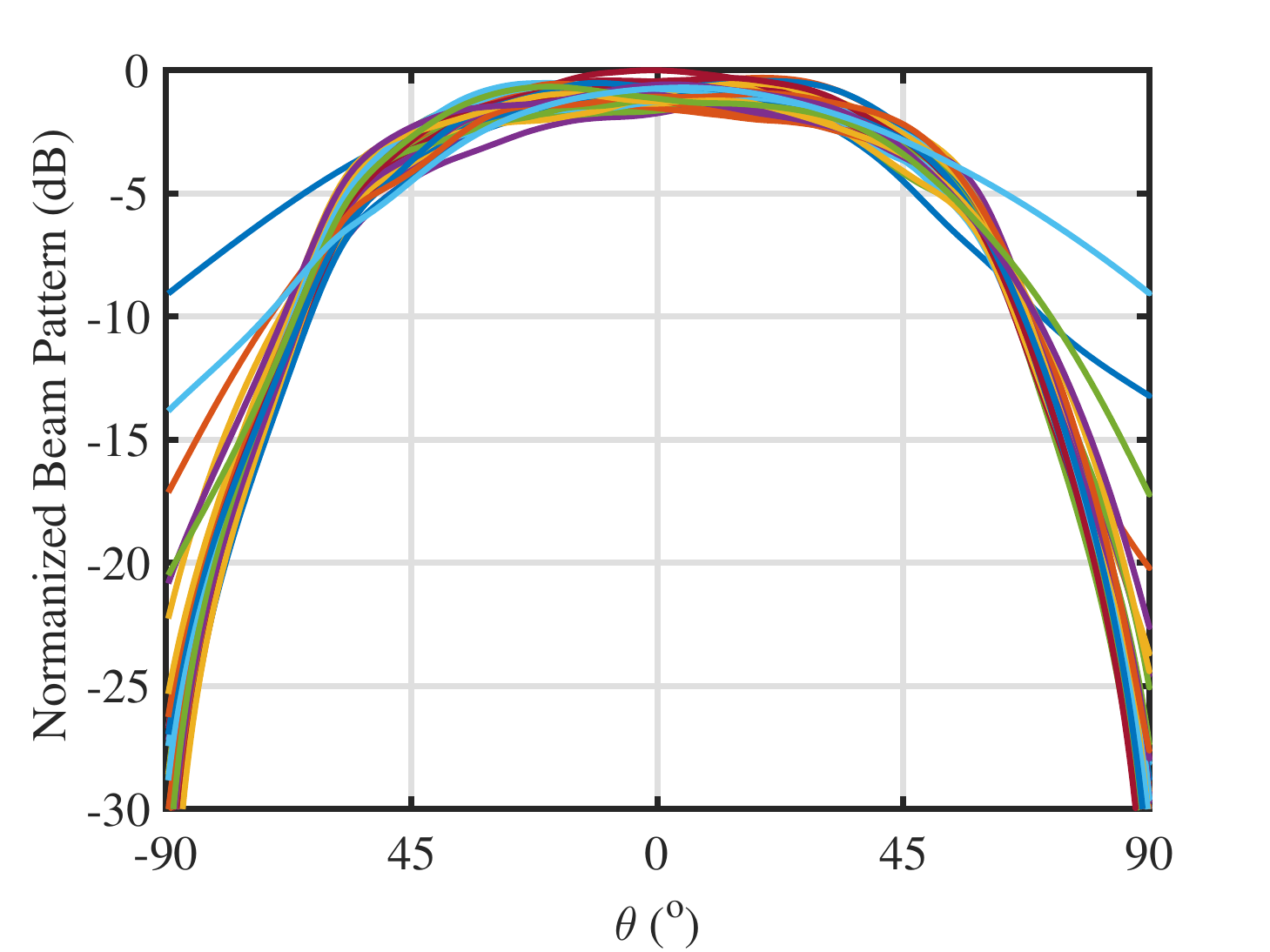}
\caption{The AEPs of the non-uniformly distributed array antenna via HFSS.} \label{fig:fig_AEPs}
\end{figure}
\begin{figure*}[!t]
\centering
\subfigure[]{\includegraphics[width=2.4in]{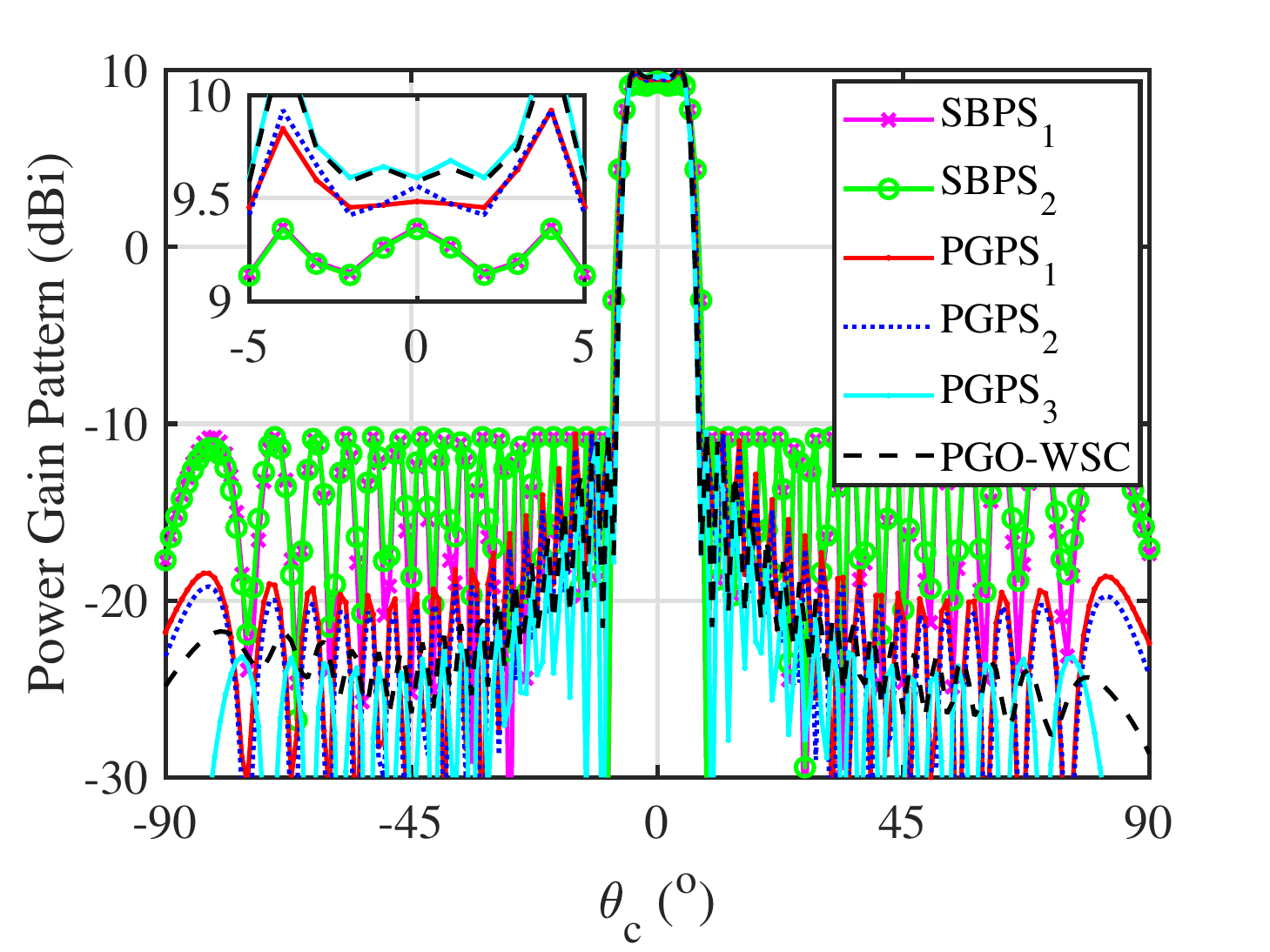}\label{fig:fig8_a}}
\subfigure[]{\includegraphics[width=2.4in]{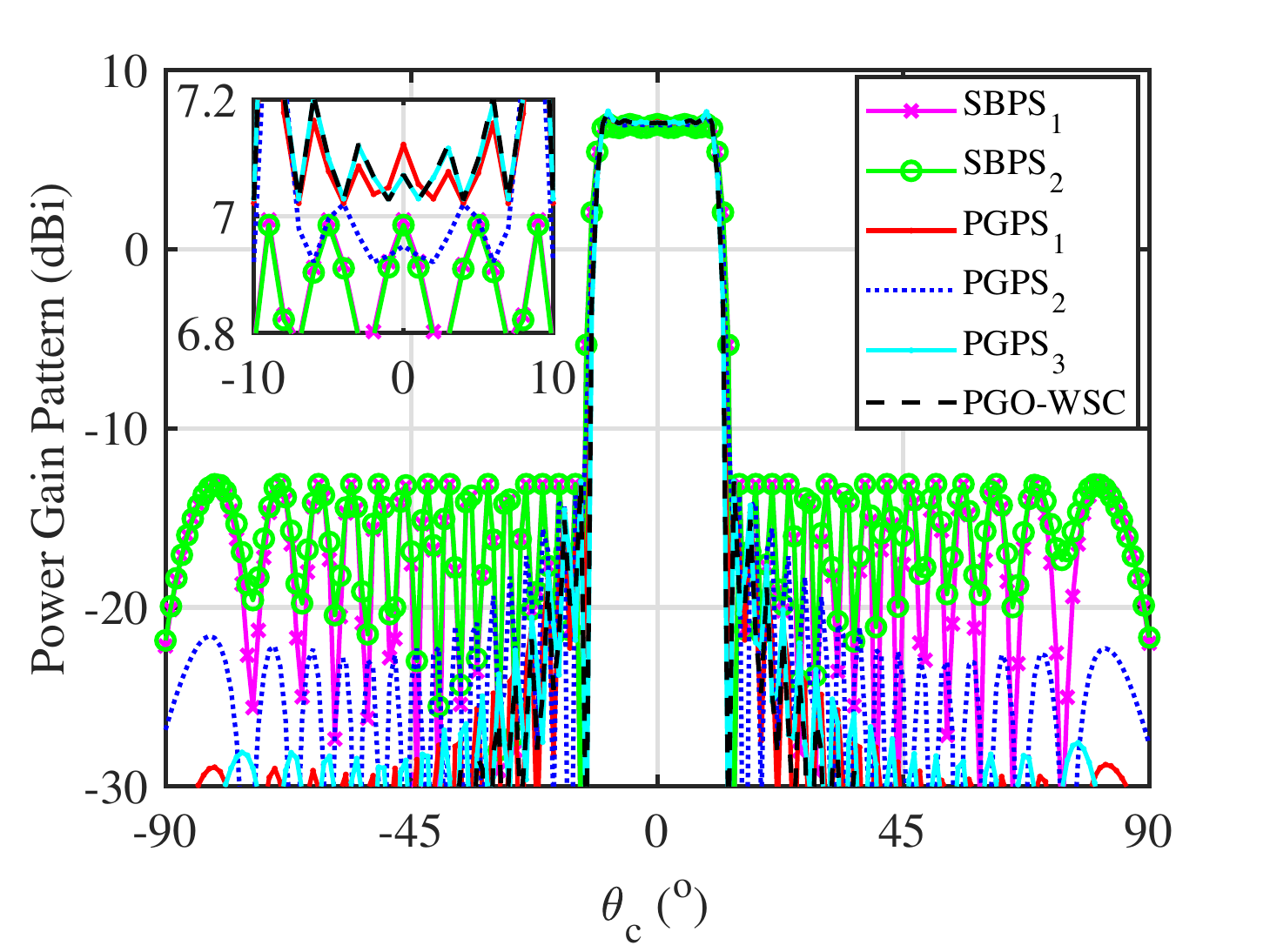}\label{fig:fig8_b}}
\subfigure[]{\includegraphics[width=2.4in]{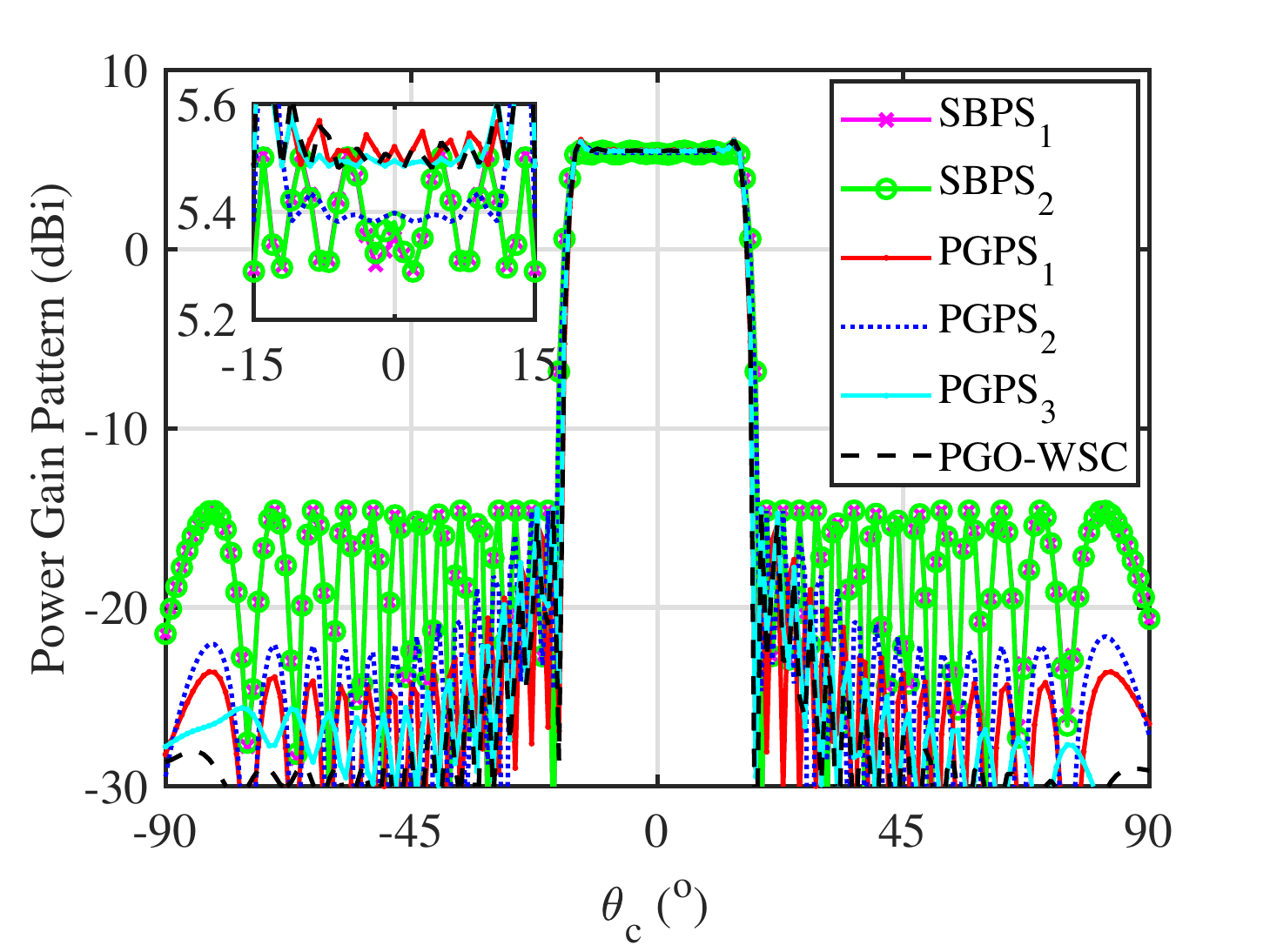}\label{fig:fig8_c}}
\subfigure[]{\includegraphics[width=2.4in]{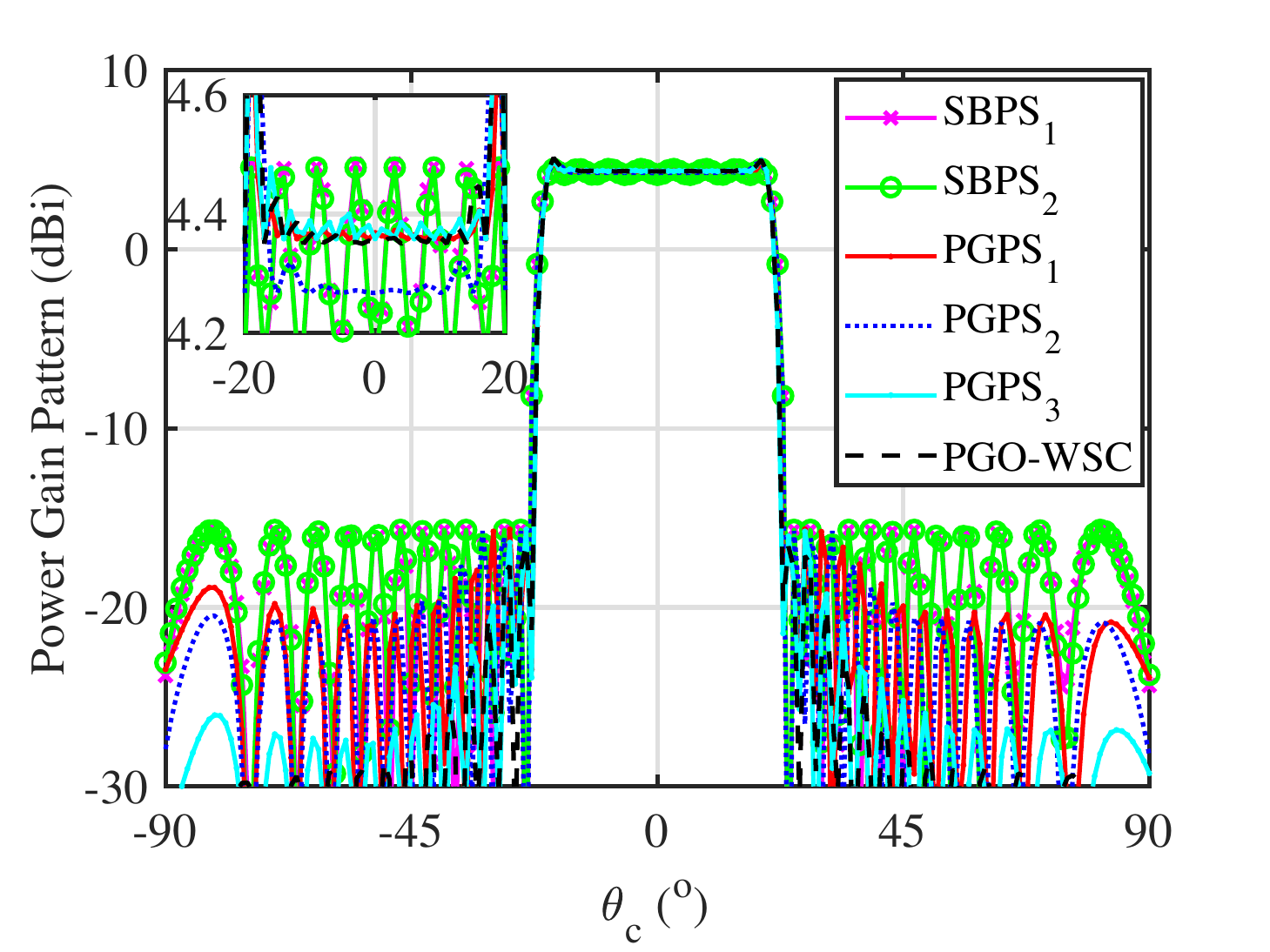}\label{fig:fig8_d}}
\caption{The synthesized power gain pattern via non-uniformly distributed AEP array with dSLL$=-20$ dB for the beamwidth: (a)  $\boldsymbol{\Theta}_\text{bw}=10^o$, (b) $\boldsymbol{\Theta}_\text{bw}=20^o$, (c) $\boldsymbol{\Theta}_\text{bw}=30^o$, and (d) $\boldsymbol{\Theta}_\text{bw}=40^o$.}
\label{fig:fig8}
\end{figure*}
With the limited distance among the elements, the mutual coupling effect always exists in practice.
To make the array synthesis more practical, the AEP is simulated as that presented in \cite{KelleyS_1993, LiuHXSYL_2017}.
The AEP is obtained by exciting one of the elements while the others are connected to matching load via High Frequency Electromagnetic Field Simulator (HFSS).
The AEPs of $41$ patch elements non-uniformly distributed linear array antenna (the elements are distributed as introduced in the previous section) is simulated and the normalized beam pattern is shown in Fig. \ref{fig:fig_AEPs}, which shows that the array elements are mainly radiating at the broadside direction with a $3$-dB beamwidth of approximately $\pm 45^o$.

Figure \ref{fig:fig8} shows the obtained power gain pattern with $\boldsymbol{\Theta}_\text{bw}=10^o$, $20^o$, $30^o$ and $40^o$ via different algorithms.
The results demonstrate that PGM-WSC obtains the highest max-minimum power gain $G_0$ in the wide-beam mainlobe as the PGPS$_3$ with the desired dSLL.
Although PGPS$_3$ obtains a similar $G_0$, it is much slower than the proposed algorithm.
Table \ref{tab:tab_AEP_SLL_G0_PSLL} shows the obtained $G_0$ and the oSLL for $\boldsymbol{\Theta}_\text{bw}=20^o$ when dSLL$=\{-20, -25, -30, -35\}$ dB.
The results also indicate that the PGM-WSC algorithm obtains higher $G_0$ to keep SLL below a desired level.
The SBPS-based algorithms perform the worst and PGPS$_1$ can obtain a similar $G_0$ as PGM-WSC does. However, PGPS$_1$ cannot obtain the predetermined SLL especially when the desired SLL is low, i.e., dSLL$=-30$ and $-35$ dB.
The other scenarios for different $\boldsymbol{\Theta}_\text{bw}$, dSLL and $\theta_\text{c}$ are also simulated with the AEP array.
Generally, all simulation results prove that the PGM-WSC obtains higher $G_0$ than other algorithms to obtain the desired SLL with a fast computation speed.
Those simulations effectively validate the superiority of the proposed PGM-WSC algorithm in maximizing the minimum power gain in wide-beam mainlobe under practical application scenarios.
\subsection{Discussions}
In the previous two parts, the performance of PGM-WoSC and PGM-WSC algorithms are validated via IEP and AEP arrays.
All simulation results prove that PGM-WoSC obtains as high $G_0$ as the PGPS-based algorithms presented in \cite{LeiYHZCQ_2019} and \cite{ZhangJZ_2020e} with guaranteed convergence. In addition, it obtains the highest $G_0$ via the PGM-WSC with the predetermined SLL and faster computational speed.
The remarks can be drawn as follows:
\begin{itemize}
\item Without the sidelobe constraint, the proposed PGM-WoSC algorithm performs as good as the existing PGPS-based algorithms, with the advantage of lower SLL in most cases.
\item With the sidelobe constraint, the proposed PGM-WSC algorithm together with other three PGPS-based algorithms presented in \cite{LeiYHZCQ_2019, LeiHCTPX_2020, ZhangJZ_2020e} always outperform the existing SBPS-based algorithms when synthesizing the wide-beam flat power gain pattern.
    Among the PGPS-based algorithms, PGM-WSC, PGPS$_1$ and PGPS$_3$ perform similarly, i.e.,  they can obtain similar $G_0$ with a difference less than $0.1$ dB.
    However, compared with PGPS$_1$, PGM-WSC can always obtain the predetermined SLL. Compared with PGPS$_3$, PGM-WSC is able to run faster.
\item Generally, the proposed two algorithms both perform better than or at least as good as the existing algorithms.
    With the aid of the ADMM framework, the power gain maximization problem can be solved faster by solving a series of subproblems through iterative processes.
    Importantly, the subproblems in each iteration is solved pseudo-analytically, which guarantees the theoretical convergence and fast computation properties.
\end{itemize}

\section{Conclusions}
This paper proposes two algorithms for solving the nonconvex power gain maximization problem without and with the SLL constraints.
The proposed algorithms reformulated the PGPS problem into an augmented Lagrangian form to simplify the original problem by updating the variables through solving a series of subproblems iteratively.
The intractable subproblems are simplified by projecting their solutions into a finite set of smaller solution spaces, resulting in pseudo-analytical solving processes.
Moreover, the convergence of the proposed algorithms is theoretically guaranteed, which ensures that the proposed algorithm can obtain higher power gain with desired SLL.
Besides, with the ADMM framework, the proposed algorithms is much faster than the other existing PGPS-based algorithms.
Numerical examples with both IEP and AEP arrays are also carried out.
The results show that the proposed algorithms can always obtain higher power gain than the existing SBPS-based algorithms, and can obtain the desired SLL with a faster speed compared with the existing PGPS-based algorithms.

\begin{appendices}
\section{Solving Variable $\nu$ in (\textsl{\ref{eq:nu_value})}}\label{app:a}

Substituting (\ref{eq:allpha}) into (\ref{eq:f3_alpha_beta}), it obtains
\begin{align} \label{eq:lag_mul_eq}
\begin{split}
f_\nu(\bbeta) &= \bbeta^T\left(\LLambda-\nu\I\right)^{-1}\LLambda\left(\LLambda-\nu\I\right)^{-1}\bbeta\\
&~~-2\bbeta^T\left(\LLambda-\nu\I\right)^{-1}\bbeta + \tilde{\d}^T\tilde{\d}\\
&=\sum_{n=1}^N\frac{\beta_n^2(2\nu-\lambda_n)}{(\lambda_n-\nu)^2}+ \tilde{\d}^T\tilde{\d}
\end{split}
\end{align}

\begin{proposition}\label{proposition_1}
 Let $\nu_1$ and $\nu_2$ be two roots of the Lagrange multiplier equation (\ref{eq:nu_value}).
 Their corresponding cost functions in (\ref{eq:lag_mul_eq}) are $f_{\nu_1}(\bbeta)$ and $f_{\nu_2}(\bbeta)$.
 If $\nu_1\geq\nu_2$, $f_{\nu_1}(\bbeta)\geq f_{\nu_2}(\bbeta)$.
\end{proposition}

\begin{proof}
Since $\nu_1$ and $\nu_2$ are the solutions of (\ref{eq:nu_value}), the following expressions of $f_{\nu_1}(\x)$ and $\f_{\nu_2}(\x)$ can be obtained as
\begin{align}
\begin{split}
f_{\nu_1}(\bbeta)&=\sum_{n=1}^N\frac{\beta_n^2}{\nu_1-\lambda_n}+\nu_1+\tilde{\d}^T\tilde{\d}\\
f_{\nu_3}(\bbeta)&=\sum_{n=1}^N\frac{\beta_n^2}{\nu_2-\lambda_n}+\nu_2+\tilde{\d}^T\tilde{\d}
\end{split}
\end{align}

Hence, there is
\begin{align}
\begin{split}
&f_{\nu_1}(\bbeta)-f_{\nu_2}(\bbeta)\\
&~~=\nu_1-\nu_2+\sum_{n=1}^N\left(\frac{\beta_n^2}{\nu_1-\lambda_n}-\frac{\beta_n^2}{\nu_2-\lambda_n}\right)\\
&~~=(\nu_1-\nu_2)\left(1-\sum_{n=1}^N\frac{\beta_2^2}{(\nu_1-\lambda_n)(\nu_2-\lambda_n)}\right)\\
&~~\geq(\nu_1-\nu_2)\left(1-\sum_{n=1}^N\frac{\beta_n^2}{2(\nu_1-\lambda_n)^2}\right.\\
&~~~~~~~~~~~~~~~~~~~~\left.-\sum_{n=1}^N\frac{\beta_n^2}{2(\nu_2-\lambda_n)^2}\right)\\
&~~=0
\end{split}
\end{align}

Hence, it yields $f_{\nu_1}(\bbeta)\geq f_{\nu_2}(\bbeta)$.
This completes the proof.
\end{proof}

With the conclusion of Proposition \ref{proposition_1}, the minimal value of $f_3$ in (\ref{eq:lag_mul_eq}) would be obtained with the minimum $\nu$, i.e. $\nu_{\min}$.

From the above euqation, there must be
\begin{align}
\begin{split}
\begin{cases}
&\max{\frac{|\beta_n|}{|\nu_{\min}-\lambda_n|}}\leq 1\\
&\max{\frac{|\beta_n|}{|\nu_{\min}-\lambda_n|}}\geq\frac{1}{\sqrt{N}}\\
&\min{\frac{|\beta_n|}{|\nu_{\min}-\lambda_n|}}\leq\frac{1}{\sqrt{N}}
\end{cases}, ~~n = 1,\cdots,N
\end{split},
\end{align}
which would yield
\begin{align} \label{eq:min_nu_domain}
\begin{split}
&\nu_{\min}\in\left[\min\left\{\lambda_n-\sqrt{N}|\beta_n|\right\}, \right.\\
&~~~\left.\min\left\{\min\left\{\lambda_n-|\beta_n|\right\},\max\left\{\lambda_n-\sqrt{N}|\beta_n|\right\}\right\}\right]
\end{split}
\end{align}

Let $f_1(\nu) = \sum_{n=1}^N\frac{\beta_n^2}{(\nu-\lambda_n)^2}-1$, the partial derivative of $f_1(\nu_{\min})=-\sum_{n=1}^N\frac{\beta_n^2}{(\nu-\lambda_n)^3}\geq0$.
It proves that function $f_1(\nu)$ is monotonically increasing in domain of (\ref{eq:min_nu_domain}).
Hence the the minimal solution, i.e., $\nu_{\min}$ of (\ref{eq:nu_value}) can be obtained with bisection searching method.

\end{appendices}





\ifCLASSOPTIONcaptionsoff
  \newpage
\fi

\bibliographystyle{IEEEbib}
\small
\bibliography{referencesAll}
\end{document}